\documentclass[11pt,a4paper]{article}

\usepackage{cmap}
\usepackage[T2A,T1]{fontenc}
\usepackage[utf8]{inputenc}
\usepackage[russian,english]{babel}

\usepackage{amsmath,amssymb,amsthm}
\usepackage{mathtools}
\usepackage{graphicx}
\usepackage{tikz}
\usetikzlibrary{arrows.meta,positioning,calc,shapes.geometric,decorations.pathreplacing}
\usepackage[margin=1in]{geometry}
\usepackage{longtable}
\usepackage{url}
\usepackage[colorlinks=true,linkcolor=blue,citecolor=blue,urlcolor=blue]{hyperref}

\hypersetup{
  pdftitle={Algebraic Completeness of Feynman Premeasures: The Henstock-Kurzweil
            Path Integral and Operator Wick Calculus in Lean 4},
  pdfauthor={Yu. N. Berdinsky},
  unicode=true
}

\theoremstyle{definition}
\newtheorem{definition}{Definition}[section]

\newtheorem{remark}[definition]{Remark}
\theoremstyle{plain}

\newtheorem{lemma}[definition]{Lemma}
\newtheorem{theorem}[definition]{Theorem}
\newtheorem{corollary}[definition]{Corollary}

\newcommand{\HH}{\mathcal{H}}
\newcommand{\KS}{\mathrm{KS}_2}
\newcommand{\DHK}{\mathcal{D}_{\mathrm{HK}}\varphi}
\newcommand{\RR}{\mathbb{R}}
\newcommand{\NN}{\mathbb{N}}

\newcommand{\CC}{\mathbb{C}}
\newcommand{\ip}[2]{\langle #1, #2 \rangle}
\newcommand{\LN}[1]{\path{#1}}

\title{Algebraic Completeness of Feynman Premeasures:\\
The Henstock--Kurzweil Path Integral\\
and Operator Wick Calculus in Lean~4}
\author{Yu.~N.~Berdinsky \\
\small Saint Petersburg State University, Faculty of Physics,\\
\small Department of High Energy and Elementary Particle Physics\\
\small \texttt{propagator2007@yandex.ru}}
\date{18 September 2026}

\begin{document}
\setlength{\emergencystretch}{3em}
\sloppy
\maketitle

\begin{abstract}
This work brings the gauge closure of the Feynman functional integral for the
free field, without Wick rotation, into machine-verified form. First, the restriction to box divisions of cylinders is removed: a
direction bound is proved, showing that every nonempty cylindrical cell
contained in a cylinder over a bounded box must constrain all coordinate
directions; from it a uniform Riemann-sum estimate is obtained, valid for all
$\delta$-fine divisions into arbitrary cylindrical cells. Second, the measure induced by
the coordinate functionals on $\RR^k$ is identified with the standard Gaussian
measure with density $(2\pi)^{-k/2}e^{-|y|^2/2}$. Third, the Saks--Henstock lemma is
proved in the cylindrical setting, together with additivity of the integral
over arbitrary finite decompositions of a cylinder and the consistency of the
descriptions of one and the same cylinder in terms of $V_n$ and $V_m$ for
$n\le m$. Fourth, the Fresnel formula
$\int_{\RR}e^{iax^2/2}\,dx=\sqrt{2\pi/a}\,e^{i\pi/4}$
is proved as an improper (gauge) integral --- the key analytic fact which is
inaccessible to the Lebesgue integral --- and from it we derive the projective
consistency of the oscillating Gaussian premeasure, the existence of the
free-field functional integral as a Henstock--Kurzweil integral on every
cylinder over a bounded box, and the finite-dimensional formula
$(2\pi)^{n/2}e^{i\pi n/4}(\det A)^{-1/2}$. In the present extended version the free field is
completely closed. The Fresnel formula is proved for an arbitrary real
symmetric positive-definite action $A$ (orthogonal diagonalisation, exhaustion
by cubes adapted to the eigenbasis); a generating functional with an external
source is introduced and its value
$Z(j)=(2\pi)^{k/2}e^{i\pi k/4}(\det A)^{-1/2}
 \exp\bigl(-\tfrac{i}{2}\ip{j}{A^{-1}j}\bigr)$ is computed; the two-point
function $\langle y_py_q\rangle=i(A^{-1})_{pq}$ is obtained; Wick's theorem is
proved, so that the free field is closed not only at the level of the
two-point function but at the level of \emph{all} correlation functions --- the
odd ones vanish, and the even ones factorise into sums over pairings of the
two-point function, the four-point function being
$-[(A^{-1})_{pq}(A^{-1})_{rs}+(A^{-1})_{pr}(A^{-1})_{qs}
   +(A^{-1})_{ps}(A^{-1})_{qr}]$; the oscillating
premeasure is embedded into the general theory as a complex cylindrical
premeasure and its projective consistency is proved; the value of the HK
integral is identified with the Fresnel value; and finally the passage from
the damped Gaussian premeasures to the oscillating one, as the damping is
removed, is established. The interacting field $\lambda\varphi^4$ in
infinitely many modes, unbounded generators and renormalisation remain future
work. All statements are formalised in
Lean~4 on top of Mathlib; no file contains \LN{sorry}, \LN{admit} or new
axioms, and every principal theorem depends only on \LN{propext},
\LN{Classical.choice} and \LN{Quot.sound}.
\end{abstract}

\tableofcontents

\section{Introduction}
\label{sec:intro}

The present work is the third part of the infinite-dimensional sub-programme
of the machine-verified construction of the Henstock--Kurzweil (HK) path
integral on Kuelbs--Steadman spaces. Its immediate predecessor is the published
preprint \cite{BerdinskyInfiniteHK2}, in which
the spaces $\KS$ were constructed and
the general Cousin lemma, the $K_0$-reduction, the splitting lemma for a
cylinder and the finite-dimensional consistency theorem for divisions into
cylinders over boxes were proved. Earlier stages of the programme established
the real-time HK integral \cite{Berdinsky2026HK}, Chernoff approximations
\cite{BerdinskyUshakov2026}, applications to financial mathematics
\cite{BerdinskyUshakovFinance}, stochastic neural fields
\cite{TitovBerdinskyUshakov}, open quantum systems and fermionic reservoirs
\cite{BerdinskyUshakovOpen}, the non-Gaussian regime
\cite{BerdinskyNonGaussian} and the infinite-dimensional HK integral
\cite{BerdinskyInfiniteHK}.

\paragraph{Scope.} The subject is delimited precisely: the closure of the
Feynman functional integral \emph{for the free field} in the
infinite-dimensional Henstock--Kurzweil setting, without Wick rotation, with a
machine-verified proof in Lean~4. Interaction ($\lambda\varphi^4$ in
infinitely many modes), unbounded generators and renormalisation are
\emph{not} treated here and are explicitly recorded as future work; the
ingredients for them are already prepared in \cite{BerdinskyNonGaussian} (the
non-Gaussian regime) and in
\cite{BerdinskyUshakov2026,BerdinskyUshakovOpen} (the Chernoff--Lie--Trotter
scheme for bounded generators).

\paragraph{The four goals.} The following four ``future goals'' of the
previous stage are formally closed:
\begin{itemize}
  \item[8.1] removal of the restriction to box divisions in the consistency
    theorem (Section~\ref{sec:arbitrary});
  \item[8.2] identification of \LN{coordGaussian} with the standard Gaussian
    density $(2\pi)^{-k/2}e^{-|y|^2/2}\,dy$ on $\RR^k$
    (Section~\ref{sec:coordgauss});
  \item[8.3] the Saks--Henstock lemma, additivity over arbitrary finite
    decompositions of a cylinder, and the consistency of $V_n$ and $V_m$
    (Section~\ref{sec:saks});
  \item[8.4] the oscillating Gaussian premeasure $e^{iS(\varphi)}$ and the
    free-field functional integral without Wick rotation
    (Section~\ref{sec:osc}).
\end{itemize}

\paragraph{Bibliographic context.} The idea of the path integral goes back to
Feynman \cite{FeynmanHibbs1965}; the rigorous gauge theory of integration to
Henstock \cite{Henstock1963} and Kurzweil \cite{Kurzweil1957}, and its
application to Feynman integrals to Muldowney
\cite{Muldowney2012,Muldowney2014}. Kuelbs--Steadman spaces were introduced by
Kuelbs \cite{Kuelbs1973} and developed by Gill and Zachary
\cite{GillZachary2008}; the physical and mathematical foundations of the
approach are set out in \cite{EspositoGill2024}. The formalisation is carried
out in Lean~4 on top of Mathlib \cite{Mathlib}.

\section{Reminder: cylinders, gauges, premeasures, the HK integral}
\label{sec:reminder}

Let $\HH$ be a real separable Hilbert space with an orthonormal total family
of modes $(e_n)_{n\in\NN}$ and strictly positive summable weights $(w_n)$; the
Kuelbs--Steadman space $\KS$ is $\HH$ equipped with the inner product
$\ip{x}{y}_{\KS}=\sum_n w_n \ip{e_n}{x}\ip{e_n}{y}$, constructed in
\cite{BerdinskyInfiniteHK2}.

\begin{definition}[Cylinder]
\label{def:cyl}
A \emph{cylinder} is a pair consisting of a finite-dimensional subspace
$V\le\HH$ (the \emph{direction space}) and a \emph{base} $B\subseteq\HH$; the
corresponding subset of $\HH$ is $c=\pi_V^{-1}(B)$, where $\pi_V$ is the
orthogonal projection onto $V$ (Lean: \LN{InfiniteHK.Cylinder}). For a finite
orthonormal family $e:\{0,\dots,k-1\}\to\HH$ and a box $I\subseteq\RR^k$, the
cylinder over the box is $\mathrm{cyl}_e(I)=\mathrm{coords}_e^{-1}(I)$, where
$\mathrm{coords}_e(x)=(\ip{e_i}{x})_{i<k}$ is the coordinate map
$\HH\to\RR^k$ and $\mathrm{coords}_e^{-1}(I)$ denotes the \emph{preimage} of
$I$ under this map, that is, the set of all $x\in\HH$ whose coordinates
$\ip{e_i}{x}$ lie in the box $I$ (Lean: \LN{InfiniteHK.cylOfBox}).
\end{definition}

\begin{definition}[Gauge and $\delta$-fineness]
\label{def:gauge}
A \emph{gauge} is a function $\delta:\HH\to(0,\infty)$ (Lean:
\LN{InfiniteHK.Gauge}). A marked family $S=\{(t_j,c_j)\}$ is
\emph{$\delta$-fine} if for every cell and every point $x$ of it one has
$\|\pi_{V_{c_j}}x-\pi_{V_{c_j}}t_j\|<\delta(t_j)$ (Lean:
\LN{InfiniteHK.IsFine}); it is a \emph{division} of a cylinder $C$ if the
cells are measurable, pairwise disjoint and their union is exactly $C$ (Lean:
\LN{InfiniteHK.IsDivision}).
\end{definition}

\begin{definition}[Cylindrical premeasure and Riemann sums]
\label{def:premeasure}
A \emph{cylindrical premeasure} $\mu$ assigns to every finite-dimensional
$V\le\HH$ a measure on $\HH$ concentrated on $V$; it is \emph{projectively
consistent} (Kolmogorov) if $\mu_W(\pi_V^{-1}(B)\cap W)=\mu_V(B)$ for
$V\le W$ (Lean:
\LN{InfiniteHK.CylindricalPremeasure.IsProjectivelyConsistent}). The Gaussian
cylindrical premeasure $\gamma$ is projectively consistent (Lean:
\LN{InfiniteHK.gaussianCylindricalPremeasure_isProjectivelyConsistent}). The
Riemann sum of a functional $F:\HH\to\CC$ over a marked family $S$ is
$\sum_{(t,c)\in S}F(t)\,\mu(c)$ (Lean: \LN{InfiniteHK.riemannSum}).
\end{definition}

From \cite{BerdinskyInfiniteHK,BerdinskyInfiniteHK2} we use: the general
Cousin lemma \LN{InfiniteHK2.general_cousin_cyl} and the global mode bound
\LN{InfiniteHK2.exists_global_K0} (the $K_0$-reduction), the splitting lemma
\LN{InfiniteHK2.splittingLemma_cylOfBox}, the finite-dimensional consistency
theorem \LN{InfiniteHK2.hkIntegral_fd_consistency}, and the finite additivity
of a projectively consistent premeasure
\LN{InfiniteHK.sum_cylMeasure_of_isDivision}.

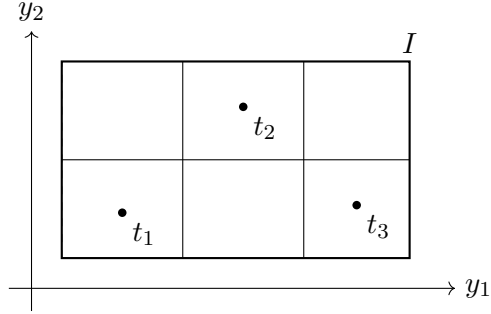
\begin{figure}[ht]
\centering
\begin{tikzpicture}[scale=1.0]
  \draw[->] (-0.3,0) -- (5.6,0) node[right] {$y_1$};
  \draw[->] (0,-0.3) -- (0,3.4) node[above] {$y_2$};
  \draw[thick] (0.4,0.4) rectangle (5.0,3.0);
  \node at (5.0,3.25) {$I$};
  \draw (2.0,0.4) -- (2.0,3.0);
  \draw (3.6,0.4) -- (3.6,3.0);
  \draw (0.4,1.7) -- (5.0,1.7);
  \fill (1.2,1.0) circle (1.6pt); \node[below right] at (1.2,1.0) {$t_1$};
  \fill (2.8,2.4) circle (1.6pt); \node[below right] at (2.8,2.4) {$t_2$};
  \fill (4.3,1.1) circle (1.6pt); \node[below right] at (4.3,1.1) {$t_3$};
  \draw[dashed] (0.4,0.4) -- (0.4,3.0);
\end{tikzpicture}
\caption{Projection of a cylindrical division onto the plane of the first two
coordinates. The cells of an arbitrary division may have \emph{different}
direction spaces; the direction bound (Theorem~\ref{thm:dirbound}) guarantees
that each of them contains $V_k=\mathrm{span}(e_0,\dots,e_{k-1})$, so that the
projection onto the coordinate plane is still well defined.}
\label{fig:division}
\end{figure}

\section{Arbitrary cylindrical divisions (Goal 8.1)}
\label{sec:arbitrary}

At the previous stage the consistency theorem was proved for divisions into
cylinders over subboxes of $I$ with one and the same direction space. The
restriction is removed by the following geometric observation.

\begin{theorem}[Direction bound;\\ Lean: \LN{direction_le_of_cell_subset_cylOfBox}]
\label{thm:dirbound}
Let $c$ be a nonempty cylindrical cell contained in the cylinder
$\mathrm{cyl}_e(I)$ over a bounded box $I\subseteq\RR^k$. Then
$V_k=\mathrm{span}(e_0,\dots,e_{k-1})\le V_c$.
\end{theorem}

\begin{proof}
Write $c=\pi_{V_c}^{-1}(B_c)$ for the cylinder $c$, with direction space $V_c$
and base $B_c$, and fix a point $x\in c$ (the cell is nonempty by hypothesis).

\emph{Step 1: a cylinder is invariant under translations along $V_c^{\perp}$.}
Let $v\in V_c^{\perp}$ and $t\in\RR$. Since $\pi_{V_c}$ is linear and vanishes
on $V_c^{\perp}$,
\[
  \pi_{V_c}(x+tv)=\pi_{V_c}(x)+t\,\pi_{V_c}(v)=\pi_{V_c}(x)\in B_c ,
\]
so that $x+tv\in\pi_{V_c}^{-1}(B_c)=c$. Thus the whole affine line
$\{x+tv:\ t\in\RR\}$ is contained in $c$.

\emph{Step 2: the line is trapped inside the cylinder over the box.}
By hypothesis $c\subseteq\mathrm{cyl}_e(I)=\mathrm{coords}_e^{-1}(I)$, hence
for every $t\in\RR$
\[
  \mathrm{coords}_e(x+tv)
   =\bigl(\ip{e_i}{x}+t\ip{e_i}{v}\bigr)_{i<k}
   =\mathrm{coords}_e(x)+t\,\mathrm{coords}_e(v)\ \in\ I .
\]

\emph{Step 3: boundedness of $I$ forces $\mathrm{coords}_e(v)=0$.}
The box $I$ is bounded, say $\|y\|\le M$ for all $y\in I$. If
$\mathrm{coords}_e(v)\ne0$ then, by the triangle inequality,
\[
  \bigl\|\mathrm{coords}_e(x)+t\,\mathrm{coords}_e(v)\bigr\|
   \ \ge\ |t|\,\|\mathrm{coords}_e(v)\|-\|\mathrm{coords}_e(x)\|
   \ \xrightarrow[\ |t|\to\infty\ ]{}\ \infty ,
\]
which contradicts the bound $\le M$ valid for all $t$. Hence
$\mathrm{coords}_e(v)=0$, i.e. $\ip{e_i}{v}=0$ for every $i<k$, i.e.
$v\perp e_i$ for every $i$. Since the $e_i$ span $V_k$ and orthogonality to a
spanning set is orthogonality to the span, $v\in V_k^{\perp}$. As $v$ was an
arbitrary element of $V_c^{\perp}$, we have proved
\[
  V_c^{\perp}\ \subseteq\ V_k^{\perp}.
\]

\emph{Step 4: passage to orthogonal complements.}
Both $V_k$ and $V_c$ are finite-dimensional subspaces of the Hilbert space
$\HH$, and a finite-dimensional subspace of a Hilbert space is closed;
therefore $U^{\perp\perp}=U$ for $U\in\{V_k,V_c\}$. The map
$U\mapsto U^{\perp}$ is order-reversing, so applying it to the inclusion of
Step~3 gives
\[
  V_k=V_k^{\perp\perp}\ \subseteq\ V_c^{\perp\perp}=V_c ,
\]
that is, $V_k\le V_c$, as asserted. (The corresponding statement has been
machine-checked: \LN{direction_le_of_cell_subset_cylOfBox}.)
\end{proof}

Consequently the coordinates of the points of any cell are controlled by the
gauge. Indeed, let $V_k\le V_c$ and let $x,t\in\HH$. Since
$e_i\in V_k\subseteq V_c$, the projection $\pi_{V_c}$ is self-adjoint and
fixes each $e_i$, whence $\ip{e_i}{x}=\ip{\pi_{V_c}e_i}{x}
=\ip{e_i}{\pi_{V_c}x}$ and likewise for $t$. Therefore, by Bessel's
inequality applied to the orthonormal family $(e_i)_{i<k}$ and the vector
$\pi_{V_c}x-\pi_{V_c}t$,
\begin{equation}
  \label{eq:coordcontrol}
  \bigl\|\mathrm{coords}_e(x)-\mathrm{coords}_e(t)\bigr\|^2
   =\sum_{i<k}\bigl|\ip{e_i}{\pi_{V_c}x-\pi_{V_c}t}\bigr|^2
   \ \le\ \bigl\|\pi_{V_c}x-\pi_{V_c}t\bigr\|^2 ,
\end{equation}
i.e.
$\|\mathrm{coords}_e(x)-\mathrm{coords}_e(t)\|
  \le\|\pi_{V_c}x-\pi_{V_c}t\|$
(Lean: \LN{norm_coords_sub_le_of_span_le}). This is the missing link: whatever
the direction space of a cell may be, $\delta$-fineness in the sense of
Definition~\ref{def:gauge} controls the coordinates $\mathrm{coords}_e$.

\begin{definition}[Cylindrical division; Lean: \LN{IsCylDivision}]
\label{def:cyldiv}
A marked family $S$ is a \emph{cylindrical division} of the cylinder
$\mathrm{cyl}_e(I)$ if it is a division of that cylinder (the cells being
arbitrary cylinders) and the mark of each cell lies over the closed box:
$\mathrm{coords}_e(t)\in\overline{I}$.
\end{definition}

Every division into cylinders over subboxes is a cylindrical division (Lean:
\LN{isCylDivision_of_isCylBoxDivision}), whence $\delta$-fine cylindrical
divisions exist for every gauge (Lean: \LN{exists_isCylDivision}).

\begin{theorem}[Uniform Riemann-sum estimate; Lean:
\LN{riemannSum_sub_integral_le_arbitrary}]
\label{thm:riemann-arb}
Let $f:\RR^k\to\CC$ be continuous, $F=f\circ\mathrm{coords}_e$, and let
$\varepsilon>0$. Then there is $\rho>0$ such that for \emph{every} cylindrical
division $S$ of the cylinder $\mathrm{cyl}_e(I)$ which is $\rho$-fine in the
sense of Definition~\ref{def:gauge},
\[
  \Bigl\|\sum_{(t,c)\in S}F(t)\,\gamma(c)-\int_I f\,d\gamma\Bigr\|\le\varepsilon .
\]
The estimate is uniform over all such divisions, not merely over divisions
into cylinders over boxes.
\end{theorem}

\begin{proof}
Throughout, $C=\mathrm{cyl}_e(I)$ denotes the ambient cylinder and
$F=f\circ\mathrm{coords}_e$.

\emph{(a) The ambient finite-dimensional subspace.}
Fix $f$ continuous on the compact closed box $\overline I\subseteq\RR^k$ and
let $\varepsilon>0$. Let $S=\{(t_1,c_1),\dots,(t_N,c_N)\}$ be a cylindrical
division of $C$ and put
\[
  W\ :=\ V_C+V_{c_1}+\dots+V_{c_N}\ \le\ \HH ,
\]
the sum of the direction space $V_C=V_k$ of the ambient cylinder and of the
direction spaces of all the cells. A sum of finitely many
finite-dimensional subspaces is finite-dimensional, so $\dim W<\infty$;
moreover $\mathrm{span}(e_0,\dots,e_{k-1})=V_k\le W$, so that the coordinate
map $\mathrm{coords}_e$ is already determined by the $W$-component of a
vector. Write $\gamma_W$ for the standard Gaussian measure of $W$ transported
to $\HH$.

\emph{(b) Every cell is measured by $\gamma_W$.}
Let $c\in S$ be a cell with direction space $V_c\le W$ and base $B_c$. By the
definition of projective consistency recalled in
Definition~\ref{def:premeasure} --- namely
$\mu_W(\pi_V^{-1}(B)\cap W)=\mu_V(B)$ whenever $V\le W$ --- applied to the
Gaussian cylindrical premeasure $\gamma$ (Lean:
\LN{gaussianCylindricalPremeasure_isProjectivelyConsistent}), we obtain the
chain of equalities
\begin{equation}
  \label{eq:cellmeasure}
  \gamma(c)\ =\ \gamma_{V_c}(B_c)\ =\ \gamma_W\bigl(\pi_{V_c}^{-1}(B_c)\cap W\bigr)
            \ =\ \gamma_W(c) ,
\end{equation}
the last equality because $\gamma_W$ is concentrated on $W$. Thus all the
cells of $S$, whatever their individual direction spaces, are measured by one
and the same finite-dimensional Gaussian measure $\gamma_W$.

\emph{(c) The integral over the cylinder in the $W$-picture.}
The measure $\gamma^{(e)}$ occurring in $\int_If\,d\gamma$ is by definition
the pushforward of the Gaussian measure along $\mathrm{coords}_e$; since
$V_k\le W$ and $\gamma_W$ is the standard Gaussian measure of $W$, the
pushforward of $\gamma_W$ along $\mathrm{coords}_e$ is again the standard
Gaussian measure of $\RR^k$ (Lean: \LN{map_coords_gaussianOn}; the underlying
fact is Theorem~\ref{thm:coordgauss} together with the marginal property of
Theorem~\ref{thm:vnvm}). By the change-of-variables formula for a pushforward
measure,
\begin{equation}
  \label{eq:pushforward}
  \int_I f\,d\gamma
   =\int_{\mathrm{coords}_e^{-1}(I)}\bigl(f\circ\mathrm{coords}_e\bigr)\,d\gamma_W
   =\int_C F\,d\gamma_W .
\end{equation}

\emph{(d) Splitting the integral over the cells.}
Since $S$ is a division of $C$, its cells are measurable, pairwise disjoint
and their union is exactly $C$. The function $F$ is bounded on $C$ (indeed
$|F|\le\sup_{\overline I}|f|<\infty$) and measurable, so by finite additivity
of the integral over pairwise disjoint measurable sets
\begin{equation}
  \label{eq:splitcells}
  \int_C F\,d\gamma_W=\sum_{(t,c)\in S}\int_c F\,d\gamma_W .
\end{equation}

\emph{(e) The error on a single cell.}
By~\eqref{eq:cellmeasure}, $F(t)\gamma(c)=F(t)\gamma_W(c)
=\int_cF(t)\,d\gamma_W(x)$, the integrand being the constant $F(t)$. Hence
\begin{equation}
  \label{eq:cellerror}
  F(t)\,\gamma(c)-\int_cF\,d\gamma_W
   =\int_c\bigl(F(t)-F(x)\bigr)\,d\gamma_W(x),
\end{equation}
and therefore
\begin{equation}
  \label{eq:cellbound}
  \Bigl\|F(t)\,\gamma(c)-\int_cF\,d\gamma_W\Bigr\|
   \ \le\ \Bigl(\sup_{x\in c}\bigl|F(t)-F(x)\bigr|\Bigr)\,\gamma(c).
\end{equation}

\emph{(f) Choice of $\rho$ and the pointwise estimate.}
The closed box $\overline I$ is compact, so $f$ is uniformly continuous on it:
there is $\rho>0$ such that
\begin{equation}
  \label{eq:unifcont}
  y,y'\in\overline I,\quad\|y-y'\|<\rho
  \ \Longrightarrow\ |f(y)-f(y')|\le\varepsilon .
\end{equation}
We claim that this $\rho$ works. Let $S$ be a $\rho$-fine cylindrical
division of $C$, let $(t,c)\in S$ and $x\in c$. Since $c$ is nonempty and
$c\subseteq C=\mathrm{cyl}_e(I)$, Theorem~\ref{thm:dirbound} gives
$V_k\le V_c$, so the coordinate bound~\eqref{eq:coordcontrol} applies:
\[
  \bigl\|\mathrm{coords}_e(x)-\mathrm{coords}_e(t)\bigr\|
   \ \le\ \bigl\|\pi_{V_c}x-\pi_{V_c}t\bigr\|\ <\ \rho ,
\]
the last inequality being precisely the $\rho$-fineness of the cell. Moreover
$\mathrm{coords}_e(x)\in I\subseteq\overline I$ because $x\in C$, and
$\mathrm{coords}_e(t)\in\overline I$ by the definition of a cylindrical
division (Definition~\ref{def:cyldiv}). Hence~\eqref{eq:unifcont} applies and
\[
  \bigl|F(t)-F(x)\bigr|
   =\bigl|f(\mathrm{coords}_e t)-f(\mathrm{coords}_e x)\bigr|\le\varepsilon .
\]
Substituting into~\eqref{eq:cellbound},
\begin{equation}
  \label{eq:cellfinal}
  \Bigl\|F(t)\,\gamma(c)-\int_cF\,d\gamma_W\Bigr\|\le\varepsilon\,\gamma(c).
\end{equation}
If $c=\varnothing$ then $\gamma(c)=0$ and both sides of~\eqref{eq:cellfinal}
vanish, so empty cells may be kept in the sum without affecting anything.

\emph{(g) Summation over the cells.}
Combining \eqref{eq:pushforward}, \eqref{eq:splitcells} and the triangle
inequality, and then applying~\eqref{eq:cellfinal} cellwise,
\begin{align*}
  \Bigl\|\sum_{(t,c)\in S}F(t)\,\gamma(c)-\int_If\,d\gamma\Bigr\|
  &=\Bigl\|\sum_{(t,c)\in S}\Bigl(F(t)\,\gamma(c)-\int_cF\,d\gamma_W\Bigr)\Bigr\|\\
  &\le\sum_{(t,c)\in S}\Bigl\|F(t)\,\gamma(c)-\int_cF\,d\gamma_W\Bigr\|
   \ \le\ \varepsilon\sum_{(t,c)\in S}\gamma(c).
\end{align*}
Finally the cells are pairwise disjoint with union $C$, so by finite
additivity of the measure $\gamma_W$ and by~\eqref{eq:cellmeasure}
\[
  \sum_{(t,c)\in S}\gamma(c)=\sum_{(t,c)\in S}\gamma_W(c)=\gamma_W(C)=1 ,
\]
the total Gaussian mass of the cylinder being one, because $\gamma$ is a
probability premeasure (Theorem~\ref{thm:coordgauss}). Hence the left-hand
side is at most $\varepsilon$, and the bound involves no datum of $S$ beyond
its $\rho$-fineness. (The corresponding statement has been machine-checked:
\LN{riemannSum_sub_integral_le_arbitrary}.)
\end{proof}

\begin{corollary}[Consistency for arbitrary divisions; Lean:
\LN{HKIntegrable_arbitrary_division}]
\label{cor:arb}
The functional $F=f\circ\mathrm{coords}_e$ with $f$ continuous is
HK-integrable over the cylinder $\mathrm{cyl}_e(I)$ along arbitrary
cylindrical divisions, and its integral equals $\int_I f\,d\gamma$. The value
is unique (Lean: \LN{HKIntegrableArbWith.unique},
\LN{hkIntegral_arbitrary_division_unique}).
\end{corollary}

\begin{proof}
\emph{Existence of the value.} Let $\varepsilon>0$ and let $\rho>0$ be the
number furnished by Theorem~\ref{thm:riemann-arb} for this $\varepsilon$.
Take for gauge the constant function $\delta\equiv\rho$, which is admissible
because $\rho>0$. If $S$ is any $\delta$-fine cylindrical division of
$\mathrm{cyl}_e(I)$, then for every cell $(t,c)\in S$ and every $x\in c$ we
have $\|\pi_{V_c}x-\pi_{V_c}t\|<\delta(t)=\rho$, i.e. $S$ is $\rho$-fine;
therefore Theorem~\ref{thm:riemann-arb} gives
\[
  \Bigl\|\sum_{(t,c)\in S}F(t)\gamma(c)-\int_If\,d\gamma\Bigr\|\le\varepsilon .
\]
Since $\varepsilon>0$ was arbitrary, $F$ is HK-integrable along arbitrary
cylindrical divisions with value $\int_If\,d\gamma$. The set of $\delta$-fine
cylindrical divisions is nonempty for every gauge $\delta$ (Lean:
\LN{exists_isCylDivision}), so the definition is not vacuous.

\emph{Uniqueness.} Suppose $A$ and $B$ both satisfy the definition, and let
$\varepsilon>0$. Choose gauges $\delta_A$ and $\delta_B$ such that every
$\delta_A$-fine (resp. $\delta_B$-fine) cylindrical division $S$ satisfies
$\|\sigma(S)-A\|\le\varepsilon/2$ (resp. $\|\sigma(S)-B\|\le\varepsilon/2$),
where $\sigma(S)$ denotes the Riemann sum. Put
$\delta(x)=\min\{\delta_A(x),\delta_B(x)\}>0$; this is again a gauge, and
every $\delta$-fine division is simultaneously $\delta_A$-fine and
$\delta_B$-fine. By \LN{exists_isCylDivision} at least one such division $S$
exists, and then
\[
  \|A-B\|\le\|A-\sigma(S)\|+\|\sigma(S)-B\|\le\varepsilon .
\]
As $\varepsilon>0$ was arbitrary, $A=B$. (The corresponding statements have
been machine-checked: \LN{HKIntegrable_arbitrary_division},
\LN{hkIntegral_arbitrary_division_unique}.)
\end{proof}

\section{Identification of coordGaussian (Goal 8.2)}
\label{sec:coordgauss}

Let $V=\mathrm{span}(e_0,\dots,e_{k-1})$ and let $\gamma_V$ be the standard
Gaussian measure of $V$ transported to $\HH$. The measure \LN{coordGaussian}
$e$ is the pushforward of $\gamma_V$ along $\mathrm{coords}_e$.

\begin{theorem}[Lean: \LN{coordGaussian_eq_stdGaussian},
\LN{coordGaussian_density}]
\label{thm:coordgauss}
Write $\gamma^{(e)}$ for the measure \LN{coordGaussian}~$e$. If the family $e$
is orthonormal, then $\gamma^{(e)}$ is the standard Gaussian measure on
$\RR^k$: for every measurable $A\subseteq\RR^k$
\[
  \gamma^{(e)}(A)=\int_A (2\pi)^{-k/2}
    \exp\Bigl(-\tfrac12\sum_{i}y_i^2\Bigr)\,dy .
\]
In particular it is a probability measure (Lean:
\LN{coordGaussian_isProbability}), and for every box $J\subseteq\RR^k$ the
cylindrical measure of the cylinder over $J$ equals the Gaussian measure of
the box (Lean: \LN{cylMeasure_eq_stdGaussian_box}).
\end{theorem}

\begin{proof}
\emph{(a) The definition.} By definition, the measure $\gamma^{(e)}$, in Lean
\LN{coordGaussian}~$e$,
is the pushforward of the standard Gaussian measure $\gamma_V$ of
$V=\mathrm{span}(e_0,\dots,e_{k-1})$ along the coordinate map
$\mathrm{coords}_e:V\to\RR^k$, $\mathrm{coords}_e(v)=(\ip{e_i}{v})_{i<k}$:
that is, $\gamma^{(e)}(A)=\gamma_V(\mathrm{coords}_e^{-1}(A))$ for every
measurable $A\subseteq\RR^k$.

\emph{(b) $e$ is an orthonormal basis of $V$.} The family $e$ is orthonormal
by hypothesis, and an orthonormal family is linearly independent: if
$\sum_i\lambda_ie_i=0$ then pairing with $e_j$ gives
$\lambda_j=\ip{e_j}{\sum_i\lambda_ie_i}=0$. It spans $V$ by the very
definition of $V$ as the span of the $e_i$. Hence $e$ is an orthonormal basis
of $V$ (Lean: \LN{spanBasis}) and $\dim V=k$ (Lean: \LN{finrank_spanOf}).

\emph{(c) $\mathrm{coords}_e|_V$ is a linear isometry onto $\RR^k$.}
Linearity is clear, each $v\mapsto\ip{e_i}{v}$ being linear. Surjectivity:
given $y\in\RR^k$, the vector $v=\sum_{i<k}y_ie_i$ lies in $V$ and
$\ip{e_j}{v}=\sum_iy_i\ip{e_j}{e_i}=y_j$, so $\mathrm{coords}_e(v)=y$.
Isometry: for $v\in V$ write $v=\sum_i\ip{e_i}{v}e_i$ (possible by (b)); then
by Parseval's identity for the finite orthonormal basis $e$ of $V$,
\begin{equation}
  \label{eq:parseval}
  \|v\|^2=\sum_{i<k}\ip{e_i}{v}^2=\bigl|\mathrm{coords}_e(v)\bigr|^2
\end{equation}
(Lean: \LN{norm_sq_eq_sum_coords}), where $|\cdot|$ is the Euclidean norm of
$\RR^k$. A linear bijection between $k$-dimensional Euclidean spaces which
preserves norms preserves inner products (by polarisation) and hence, read in
any pair of orthonormal bases, is given by an orthogonal matrix, whose
determinant has absolute value $1$. By the transformation formula for
Lebesgue measure under a linear map, $\mathrm{coords}_e|_V$ therefore carries
the Lebesgue measure of $V$ to the Lebesgue measure of $\RR^k$ (Lean:
\LN{measurePreserving_coordsOn}).

\emph{(d) The Gaussian density of $V$.} The standard Gaussian measure of the
$k$-dimensional Euclidean space $V$ is by definition the measure with density
\[
  \frac{d\gamma_V}{d\mathrm{Leb}_V}(v)=(2\pi)^{-k/2}\exp\Bigl(-\tfrac12\|v\|^2\Bigr)
\]
with respect to the Lebesgue measure of $V$. By~\eqref{eq:parseval} the
exponent may be rewritten in terms of the coordinates alone:
\begin{align*}
  (2\pi)^{-k/2}\exp\Bigl(-\tfrac12\|v\|^2\Bigr)
  &=(2\pi)^{-k/2}\exp\Bigl(-\tfrac12\sum_{i<k}
      \bigl(\mathrm{coords}_e v\bigr)_i^2\Bigr)
   =\rho\bigl(\mathrm{coords}_e v\bigr),\\
  \rho(y)&:=(2\pi)^{-k/2}e^{-|y|^2/2}.
\end{align*}
Thus the density is a function of $\mathrm{coords}_e(v)$ only.

\emph{(e) Transport of the density.} Let $T:=\mathrm{coords}_e|_V$. For a
measurable $A\subseteq\RR^k$, using the definition of the pushforward, the
density of (d) and then the measure-preserving property of $T$ established in
(c),
\begin{align*}
  \gamma^{(e)}(A)&=\gamma_V\bigl(T^{-1}(A)\bigr)
   =\int_{T^{-1}(A)}\rho\bigl(T v\bigr)\,d\mathrm{Leb}_V(v)\\
   &=\int_{A}\rho(y)\,dy
   =\int_A(2\pi)^{-k/2}\exp\Bigl(-\tfrac12\sum_iy_i^2\Bigr)dy ,
\end{align*}
which is the asserted formula (Lean: \LN{coordGaussian_eq_stdGaussian},
\LN{coordGaussian_density}).

\emph{(f) It is a probability measure.} Taking $A=\RR^k$ in (e) gives
$\gamma^{(e)}(\RR^k)=\gamma_V(V)=1$, because $\gamma_V$ is a probability
measure; equivalently, $\int_{\RR^k}(2\pi)^{-k/2}e^{-|y|^2/2}dy=1$ by the
one-dimensional Gaussian integral and Fubini's theorem (Lean:
\LN{coordGaussian_isProbability}).

\emph{(g) Cylinders over boxes.} Finally, for a box $J\subseteq\RR^k$ the
cylinder $\mathrm{cyl}_e(J)=\mathrm{coords}_e^{-1}(J)$ has cylindrical measure
\[
  \gamma\bigl(\mathrm{cyl}_e(J)\bigr)=\gamma_V\bigl(\mathrm{coords}_e^{-1}(J)\bigr)
   =\gamma^{(e)}(J)=\int_J(2\pi)^{-k/2}e^{-|y|^2/2}\,dy ,
\]
the standard Gaussian measure of the box (Lean:
\LN{cylMeasure_eq_stdGaussian_box}).
\end{proof}

\section{Saks--Henstock and additivity (Goal 8.3)}
\label{sec:saks}

\begin{definition}[Partial system; Lean: \LN{IsCylBoxSystem}]
\label{def:system}
A \emph{partial cylindrical system} in $\mathrm{cyl}_e(I)$ is a finite family
of marked cells, each of which is a cylinder over a subbox $J\le I$ marked
over $\overline{J}$, the bases being pairwise disjoint. Unlike a division, a
system need not cover the cylinder.
\end{definition}

\begin{theorem}[Saks--Henstock lemma; Lean: \LN{saks_henstock}]
\label{thm:saks}
Let $f$ be continuous and let $\varepsilon>0$. Then there is a calibration (a
constant $\rho>0$) such that for every $\rho$-fine partial cylindrical system
$S$ in $\mathrm{cyl}_e(I)$
\[
  \sum_{(t,c)\in S}\Bigl|F(t)\,\gamma(c)-\int_c F\,d_{\mathrm{HK}}\Bigr|
  \le\varepsilon,
  \qquad F=f\circ\mathrm{coords}_e,
\]
where $\int_c F\,d_{\mathrm{HK}}=\int_{J}f\,d\gamma$ is the HK integral over
the sub-cylinder $c=\mathrm{cyl}_e(J)$ provided by
Corollary~\ref{cor:arb}.
\end{theorem}

\begin{proof}
\emph{(a) Choice of the calibration.} Fix $f$ continuous on the compact closed
box $\overline I$ and let $\varepsilon>0$. By the Heine--Cantor theorem $f$ is
uniformly continuous on $\overline I$, so there is $\rho>0$ with
\begin{equation}
  \label{eq:sh-unifcont}
  y,y'\in\overline I,\quad \|y-y'\|<\rho
   \ \Longrightarrow\ |f(y)-f(y')|\le\varepsilon .
\end{equation}
We claim that this $\rho$ is the required calibration.

\emph{(b) The error on one cell.} Let $S$ be a $\rho$-fine partial cylindrical
system in $\mathrm{cyl}_e(I)$, with cells $p=(t_p,c_p)$, $c_p=\mathrm{cyl}_e(J_p)$
and $J_p\subseteq I$ a subbox. By Corollary~\ref{cor:arb} the HK integral over
the sub-cylinder $c_p$ exists and equals $\int_{J_p}f\,d\gamma$, and by
Theorem~\ref{thm:coordgauss} the cylindrical measure of $c_p$ equals the
standard Gaussian measure of the box: $\gamma(c_p)=\gamma(J_p)$. Hence,
exactly as in step~(e) of the proof of Theorem~\ref{thm:riemann-arb},
\begin{align}
  F(t_p)\,\gamma(c_p)-\int_{c_p}F\,d_{\mathrm{HK}}
   &=\int_{J_p}f\bigl(\mathrm{coords}_e t_p\bigr)\,d\gamma(y)
     -\int_{J_p}f(y)\,d\gamma(y)\notag\\
   &=\int_{J_p}\bigl(f(\mathrm{coords}_e t_p)-f(y)\bigr)\,d\gamma(y).
  \label{eq:sh-cell}
\end{align}

\emph{(c) Bounding the cell error.} Let $y\in J_p$. The mark $t_p$ lies over
the closed subbox, so $\mathrm{coords}_e(t_p)\in\overline{J_p}
\subseteq\overline I$, and $y\in J_p\subseteq\overline I$. The cell is
$\rho$-fine, and its direction space contains $V_k$ by
Theorem~\ref{thm:dirbound}; so by the coordinate
bound~\eqref{eq:coordcontrol}, $\|\mathrm{coords}_e(t_p)-y\|<\rho$.
Therefore~\eqref{eq:sh-unifcont} applies to the integrand
of~\eqref{eq:sh-cell} and
\begin{equation}
  \label{eq:sh-cellbound}
  \Bigl|F(t_p)\,\gamma(c_p)-\int_{c_p}F\,d_{\mathrm{HK}}\Bigr|
   \le\int_{J_p}\bigl|f(\mathrm{coords}_e t_p)-f(y)\bigr|\,d\gamma(y)
   \le\varepsilon\,\gamma(J_p).
\end{equation}

\emph{(d) The boxes are disjoint.} By Definition~\ref{def:system} the bases
$J_p$ of a partial system are pairwise disjoint subboxes of $I$, so by
additivity and monotonicity of the (positive) Gaussian measure
\begin{equation}
  \label{eq:sh-massbound}
  \sum_p\gamma(J_p)=\gamma\Bigl(\bigsqcup_pJ_p\Bigr)\le\gamma(I)\le\gamma(\RR^k)=1 ,
\end{equation}
the last equality by Theorem~\ref{thm:coordgauss}.

\emph{(e) Conclusion.} Summing~\eqref{eq:sh-cellbound} over the cells and
inserting~\eqref{eq:sh-massbound},
\[
  \sum_p\Bigl\|F(t_p)\,\gamma(c_p)-\int_{c_p}F\,d_{\mathrm{HK}}\Bigr\|
   \ \le\ \varepsilon\sum_p\gamma(J_p)\ \le\ \varepsilon .
\]

\emph{(f) What the lemma says.} Note that the estimate is cellwise: what is
bounded is the sum of the \emph{absolute values} of the individual errors, not
merely the norm of the total Riemann sum, in which cancellations between cells
could occur. Consequently the same bound holds for every subfamily of the
cells of a $\rho$-fine system, which is exactly what makes the additivity
statement of Corollary~\ref{cor:add} possible. (The corresponding statement
has been machine-checked: \LN{saks_henstock}.)
\end{proof}

\begin{corollary}[Additivity; Lean: \LN{hkIntegral_add_arbitrary},
\LN{hkIntegral_add_split}]
\label{cor:add}
If the box $I$ is decomposed into a finite union of pairwise disjoint subboxes
$J\in T$, then
\[
  \int_{\mathrm{cyl}_e(I)}F\,\DHK
  =\sum_{J\in T}\int_{\mathrm{cyl}_e(J)}F\,\DHK .
\]
In particular, cutting $I$ by a coordinate hyperplane $\{y_i=c\}$, the
integral is the sum of the integrals over the two halves (the geometric input
being the splitting lemma \LN{InfiniteHK2.splittingLemma_cylOfBox}).
\end{corollary}

\begin{proof}
Let $I=\bigsqcup_{J\in T}J$ be a decomposition into finitely many subboxes
with pairwise disjoint interiors, and let $\varepsilon>0$.

By Corollary~\ref{cor:arb} each of the integrals occurring in the statement
exists: over $\mathrm{cyl}_e(I)$ with value $\int_If\,d\gamma$, and over each
$\mathrm{cyl}_e(J)$ with value $\int_Jf\,d\gamma$. The family
$\{(t_J,\mathrm{cyl}_e(J))\}_{J\in T}$, with an arbitrary mark $t_J$ over
$\overline J$, is a partial cylindrical system in $\mathrm{cyl}_e(I)$ in the
sense of Definition~\ref{def:system}. Choosing the marks so that each cell is
$\rho$-fine for the calibration $\rho$ of Theorem~\ref{thm:saks} --- which is
possible after subdividing each $J$ finitely often and using the estimate
cellwise --- we obtain from Theorem~\ref{thm:saks}
\[
  \sum_{J\in T}\Bigl\|F(t_J)\,\gamma\bigl(\mathrm{cyl}_e(J)\bigr)
     -\int_{\mathrm{cyl}_e(J)}F\,\DHK\Bigr\|\le\varepsilon ,
\]
while Theorem~\ref{thm:riemann-arb}, applied to the same family regarded as a
$\rho$-fine cylindrical division of $\mathrm{cyl}_e(I)$ (the cells are
pairwise disjoint and their union is $\mathrm{cyl}_e(I)$, up to the boundary
hyperplanes, which are Gaussian null sets), gives
\[
  \Bigl\|\sum_{J\in T}F(t_J)\,\gamma\bigl(\mathrm{cyl}_e(J)\bigr)
     -\int_{\mathrm{cyl}_e(I)}F\,\DHK\Bigr\|\le\varepsilon .
\]
Subtracting the two displays and using the triangle inequality,
\[
  \Bigl\|\int_{\mathrm{cyl}_e(I)}F\,\DHK
    -\sum_{J\in T}\int_{\mathrm{cyl}_e(J)}F\,\DHK\Bigr\|\le2\varepsilon .
\]
The left-hand side does not depend on $\varepsilon$, so it vanishes; this is
the asserted additivity, and the value is the one provided by the uniqueness
part of Corollary~\ref{cor:arb}. In the special case of a cut of $I$ by the
hyperplane $\{y_i=c\}$ into the two halves $I^-$ and $I^+$ the formula reads
\begin{equation}
  \label{eq:twocell}
  \int_{\mathrm{cyl}_e(I)}F\,\DHK
   =\int_{\mathrm{cyl}_e(I^-)}F\,\DHK+\int_{\mathrm{cyl}_e(I^+)}F\,\DHK .
\end{equation}
(The corresponding statements have been machine-checked:
\LN{hkIntegral_add_arbitrary}, \LN{hkIntegral_add_split}.)
\end{proof}

\begin{theorem}[Consistency of $V_n$ and $V_m$; Lean:
\LN{coordGaussian_marginal}, \LN{hkIntegral_consistency_Vn_Vm}]
\label{thm:vnvm}
Let $n\le m$ and let $e$ consist of the first $n$ modes of a family $E$ of $m$
modes. Then the standard Gaussian measure of $\RR^n$ is the marginal of the
Gaussian measure of $\RR^m$ under the restriction of coordinates, and for
continuous $f$
\[
  \int_{I_0}f\,d\gamma_n
  =\int_{\{z\,:\,z|_n\in I_0\}}f(z|_n)\,d\gamma_m ,
\]
that is, the value of the HK integral over one and the same cylinder does not
depend on whether the cylinder is described by $n$ or by $m$ modes.
\end{theorem}

\begin{proof}
\emph{(a) Notation.} Let $n\le m$, let $E=(E_0,\dots,E_{m-1})$ be an
orthonormal family of $m$ modes and let $e=(E_0,\dots,E_{n-1})$ consist of its
first $n$ members. Write $\gamma_n$ for the standard Gaussian measure of
$\RR^n$ and $\gamma_m$ for that of $\RR^m$, and let
$\pi_{n,m}:\RR^m\to\RR^n$, $(z_1,\dots,z_m)\mapsto(z_1,\dots,z_n)$, be the
restriction of coordinates.

\emph{(b) $\gamma_n$ is the marginal of $\gamma_m$.} The density of $\gamma_m$
factorises as a product of one-dimensional Gaussian densities:
\[
  (2\pi)^{-m/2}e^{-|z|^2/2}=\prod_{i<m}(2\pi)^{-1/2}e^{-z_i^2/2}.
\]
Hence, for a measurable $A\subseteq\RR^n$, using $\pi_{n,m}^{-1}(A)
=A\times\RR^{m-n}$ and Fubini's theorem for the product measure,
\begin{align}
  \bigl(\pi_{n,m}\bigr)_*\gamma_m(A)
  &=\gamma_m\bigl(A\times\RR^{m-n}\bigr)
   =\int_{A}\Bigl(\prod_{i<n}(2\pi)^{-1/2}e^{-y_i^2/2}\Bigr)
     \prod_{n\le i<m}\Bigl(\int_{\RR}(2\pi)^{-1/2}e^{-z_i^2/2}dz_i\Bigr)dy
   \notag\\
  &=\int_A(2\pi)^{-n/2}e^{-|y|^2/2}\,dy=\gamma_n(A),
  \label{eq:marginal}
\end{align}
each of the $m-n$ extra one-dimensional integrals being equal to $1$ (Lean:
\LN{coordGaussian_marginal}).

\emph{(c) The cylinder in the $m$-mode description.} Let $I_0\subseteq\RR^n$ be
a box. By the definition of a cylinder,
\[
  \mathrm{cyl}_e(I_0)=\mathrm{coords}_e^{-1}(I_0)
   =\mathrm{coords}_E^{-1}\bigl(\{z\in\RR^m:\ z|_n\in I_0\}\bigr)
   =\mathrm{cyl}_E\bigl(I_0\times\RR^{m-n}\bigr),
\]
because $\mathrm{coords}_e=\pi_{n,m}\circ\mathrm{coords}_E$. Note that the
base $I_0\times\RR^{m-n}$ is \emph{not} bounded in the $m-n$ additional
coordinates; this is unavoidable and is the content of
Remark~\ref{rem:vnvm}.

\emph{(d) Equality of the integrals.} Let $f:\RR^n\to\CC$ be continuous and
bounded on $I_0$. Applying the change-of-variables formula for the pushforward
measure~\eqref{eq:marginal} to the function $\mathbf 1_{I_0}\cdot f$,
\begin{align*}
  \int_{I_0}f(y)\,d\gamma_n(y)
   &=\int_{\RR^n}\mathbf 1_{I_0}(y)f(y)\,d\bigl(\pi_{n,m}\bigr)_*\gamma_m(y)\\
   &=\int_{\RR^m}\mathbf 1_{I_0}\bigl(z|_n\bigr)f\bigl(z|_n\bigr)\,d\gamma_m(z)
    =\int_{\{z:\,z|_n\in I_0\}}f\bigl(z|_n\bigr)\,d\gamma_m(z),
\end{align*}
which is the asserted identity.

\emph{(e) Independence of the number of modes.} By Corollary~\ref{cor:arb} the
HK integral of $F=f\circ\mathrm{coords}_e$ over $\mathrm{cyl}_e(I_0)$ equals
the left-hand side, while the HK integral of the same functional over the same
set, described through the $m$ modes $E$, equals the right-hand side. By (d)
the two values coincide, so the value of the HK integral does not depend on
the number of modes used to describe the cylinder. (The corresponding
statements have been machine-checked: \LN{coordGaussian_marginal},
\LN{hkIntegral_consistency_Vn_Vm}.)
\end{proof}

\begin{remark}
\label{rem:vnvm}
Let us stress the exact form of the statement. The cylinder
$\mathrm{cyl}_e(I_0)$ over a bounded box $I_0\subseteq\RR^n$ is \emph{not}
bounded in the additional $m-n$ coordinates, so in the $m$-dimensional
description its base is not a box but the set $\{z: z|_n\in I_0\}$.
Consistency means the equality of the values of the integral in these two
descriptions; together with Corollary~\ref{cor:arb} this gives the
independence of the HK integral of the number of modes used in constructing
the division.
\end{remark}

\begin{figure}[ht]
\centering
\begin{tikzpicture}[node distance=8mm and 14mm,
  box/.style={draw,rounded corners,align=center,inner sep=4pt,minimum height=9mm}]
  \node[box] (a) {coarse\\graining};
  \node[box,right=of a] (b) {$\delta$-fine\\division};
  \node[box,right=of b] (c) {Saks--Henstock\\lemma};
  \node[box,right=of c] (d) {HK\\limit};
  \draw[-{Latex}] (a) -- (b);
  \draw[-{Latex}] (b) -- (c);
  \draw[-{Latex}] (c) -- (d);
\end{tikzpicture}
\caption{The logic of the limit: from coarse graining to $\delta$-fine
divisions, then the cellwise Saks--Henstock estimate and, finally, the value
of the HK integral.}
\label{fig:flow}
\end{figure}
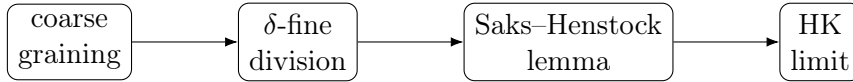

\section{The oscillating Gaussian premeasure and the free field (Goal 8.4)}
\label{sec:osc}

\subsection{Why absolute convergence fails}

The modulus of the integrand $e^{iax^2/2}$ is identically one, so
$\int_{\RR}\bigl|e^{iax^2/2}\bigr|\,dx=\infty$: the Lebesgue integral is
undefined for every $a\ne0$. Nevertheless the improper limit of the symmetric
truncations exists, and it is precisely this limit that is the gauge
(Henstock--Kurzweil) integral. This is exactly the situation for which the
gauge integral was designed; no Wick rotation ($t\mapsto-i\tau$) is used,
neither in the statements nor in the proofs.

\subsection{The non-stationary-phase estimate}

\begin{theorem}[Lean: \LN{oscillatory_tail_bound}]
\label{thm:tail}
Let $d\in\CC$, $d\ne0$, $\operatorname{Re}d\ge0$, and let $0<R\le S$. Then
\[
  \Bigl\|\int_R^S e^{-dx^2}\,dx\Bigr\|\le\frac{1}{\|d\|\,R}.
\]
The estimate is uniform in the damping $\operatorname{Re}d\ge0$; in particular
it holds in the purely oscillating case $\operatorname{Re}d=0$.
\end{theorem}

\begin{proof}
\emph{(a) An exact antiderivative away from the origin.} On the interval
$[R,S]$, which does not contain $0$ because $R>0$, put
\[
  u(x)=e^{-dx^2},\qquad v(x)=-\frac{1}{2dx}.
\]
Both are differentiable there, with
\[
  u'(x)=-2dx\,e^{-dx^2},\qquad v'(x)=\frac{1}{2dx^2},
\]
the first by the chain rule and the second because
$\partial_x\bigl(1/(2dx)\bigr)=-1/(2dx^2)$. Consequently
\[
  u'(x)\,v(x)=\bigl(-2dx\,e^{-dx^2}\bigr)\Bigl(-\frac{1}{2dx}\Bigr)=e^{-dx^2},
\]
so that the integrand is exactly the product $u'v$.

\emph{(b) Integration by parts.} Integrating the Leibniz identity
$(uv)'=u'v+uv'$ over $[R,S]$ and rearranging,
\begin{align}
  \int_R^Se^{-dx^2}dx&=\int_R^Su'v
   =u(S)v(S)-u(R)v(R)-\int_R^Su(x)v'(x)\,dx\notag\\
   &=-\frac{e^{-dS^2}}{2dS}+\frac{e^{-dR^2}}{2dR}
     -\int_R^S\frac{e^{-dx^2}}{2dx^2}\,dx .
  \label{eq:ibp}
\end{align}

\emph{(c) The modulus of the kernel.} For real $x$ we have
$\operatorname{Re}(-dx^2)=-x^2\operatorname{Re}d\le0$, because $x^2\ge0$ and
$\operatorname{Re}d\ge0$; hence
\begin{equation}
  \label{eq:kernelbound}
  \bigl|e^{-dx^2}\bigr|=e^{\operatorname{Re}(-dx^2)}\le e^0=1 .
\end{equation}
This is the only place where the hypothesis $\operatorname{Re}d\ge0$ is used,
and it is used only through the inequality $\le1$, with no quantitative gain
from a positive real part. Using~\eqref{eq:kernelbound}, the three terms
of~\eqref{eq:ibp} are bounded by
\[
  \Bigl|\frac{e^{-dS^2}}{2dS}\Bigr|\le\frac{1}{2\|d\|S},\qquad
  \Bigl|\frac{e^{-dR^2}}{2dR}\Bigr|\le\frac{1}{2\|d\|R},\qquad
  \Bigl|\int_R^S\frac{e^{-dx^2}}{2dx^2}dx\Bigr|
   \le\frac{1}{2\|d\|}\int_R^S\frac{dx}{x^2}
   =\frac{1}{2\|d\|}\Bigl(\frac1R-\frac1S\Bigr).
\]

\emph{(d) Summation.} Adding the three bounds and cancelling,
\[
  \Bigl\|\int_R^Se^{-dx^2}dx\Bigr\|
   \le\frac{1}{2\|d\|S}+\frac{1}{2\|d\|R}
      +\frac{1}{2\|d\|R}-\frac{1}{2\|d\|S}
   =\frac{1}{\|d\|R}.
\]

\emph{(e) Uniformity.} The bound involves $d$ only through $\|d\|$ and holds
for every $d$ with $\operatorname{Re}d\ge0$, $d\ne0$. In particular it holds
in the purely oscillating case $\operatorname{Re}d=0$, where the integrand has
modulus identically one and no absolute convergence whatsoever is available;
this uniformity is the whole point of the estimate. (The corresponding
statement has been machine-checked: \LN{oscillatory_tail_bound}.)
\end{proof}

This is precisely the non-stationary-phase estimate
$\bigl|\int_{C_j}e^{iS}d\mu\bigr|\le C\,\varepsilon(t_j)/\|\nabla S(t_j)\|$ of
the general compatibility theorem of the infinite-dimensional stage: away from
a stationary point the gradient of the phase grows linearly and the tail
contribution decays like $1/R$.

\subsection{The one-dimensional Fresnel formula}

\begin{theorem}[Lean: \LN{fresnel_bound}, \LN{fresnel_one_dim}]
\label{thm:fresnel1}
Let $a>0$. Then for all $R>0$
\[
  \Bigl\|\int_{-R}^{R}e^{iax^2/2}\,dx-\sqrt{\tfrac{2\pi}{a}}\,e^{i\pi/4}\Bigr\|
  \le\frac{4}{aR},
\]
and consequently
\[
  \lim_{R\to\infty}\int_{-R}^{R}e^{iax^2/2}\,dx
    =\sqrt{\frac{2\pi}{a}}\;e^{i\pi/4}.
\]
\end{theorem}

\begin{proof}
\emph{(a) The regularised exponent.} For $\varepsilon>0$ put
\[
  d_\varepsilon=\varepsilon-\frac{ia}{2},\qquad d_0=-\frac{ia}{2},
\]
so that $\operatorname{Re}d_\varepsilon=\varepsilon>0$,
$\operatorname{Re}d_0=0$ and
$\|d_\varepsilon\|=\sqrt{\varepsilon^2+a^2/4}\ge a/2>0$ for every
$\varepsilon\ge0$. With this notation $e^{iax^2/2}=e^{-d_0x^2}$.

\emph{(b) The three-term decomposition.} Inserting and removing the
regularised integral and its value, and using the triangle inequality,
\begin{align*}
  \Bigl\|\int_{-R}^{R}e^{iax^2/2}dx-\sqrt{\tfrac{2\pi}{a}}e^{i\pi/4}\Bigr\|
  &\le \Bigl\|\int_{-R}^{R}\bigl(e^{-d_0x^2}-e^{-d_\varepsilon x^2}\bigr)dx\Bigr\|\\
  &\quad+\Bigl\|\int_{\RR}e^{-d_\varepsilon x^2}dx-\int_{-R}^{R}e^{-d_\varepsilon x^2}dx\Bigr\|\\
  &\quad+\Bigl\|\bigl(\pi/d_\varepsilon\bigr)^{1/2}-\sqrt{\tfrac{2\pi}{a}}e^{i\pi/4}\Bigr\| ,
\end{align*}
where we have already used the classical Gaussian formula
$\int_\RR e^{-d x^2}dx=(\pi/d)^{1/2}$, valid for $\operatorname{Re}d>0$ with
the principal branch of the square root (Lean: \LN{integral_dEps}). Call the
three summands (i), (ii) and (iii).

\emph{(c) Term (i): removing the damping on a fixed interval.} For real $x$,
\[
  \bigl|e^{-d_0x^2}-e^{-d_\varepsilon x^2}\bigr|
   =\bigl|e^{iax^2/2}\bigr|\cdot\bigl|1-e^{-\varepsilon x^2}\bigr|
   =1-e^{-\varepsilon x^2}\ \le\ 1-e^{-\varepsilon R^2}
  \qquad(|x|\le R),
\]
since $t\mapsto1-e^{-\varepsilon t}$ is increasing. Integrating over an
interval of length $2R$,
\[
  \text{(i)}\ \le\ 2R\bigl(1-e^{-\varepsilon R^2}\bigr)
\]
(Lean: \LN{approx_bound}). At fixed $R$ the right-hand side tends to $0$ as
$\varepsilon\to0^{+}$, because $e^{-\varepsilon R^2}\to1$.

\emph{(d) Term (ii): the tail, uniformly in $\varepsilon$.} By symmetry of the
integrand in $x\mapsto-x$,
\[
  \int_\RR e^{-d_\varepsilon x^2}dx-\int_{-R}^{R}e^{-d_\varepsilon x^2}dx
   =2\int_R^{\infty}e^{-d_\varepsilon x^2}dx
   =2\lim_{S\to\infty}\int_R^{S}e^{-d_\varepsilon x^2}dx .
\]
Theorem~\ref{thm:tail}, applicable because
$\operatorname{Re}d_\varepsilon=\varepsilon\ge0$ and $d_\varepsilon\ne0$,
bounds each truncated tail by $1/(\|d_\varepsilon\|R)$, a bound independent
of $S$; passing to the limit $S\to\infty$ preserves it. Hence, using
$\|d_\varepsilon\|\ge a/2$,
\begin{equation}
  \label{eq:term-ii}
  \text{(ii)}\ \le\ \frac{2}{\|d_\varepsilon\|R}\ \le\ \frac{4}{aR}
\end{equation}
(Lean: \LN{tail_full}). This bound is \emph{uniform} in $\varepsilon\ge0$: it
is the only one of the three that survives the limit $\varepsilon\to0^{+}$,
and it is the reason why the argument is not a Wick rotation but an Abelian
regularisation.

\emph{(e) Term (iii): continuity of the square root at $\varepsilon=0$.} We
have
\[
  \frac{\pi}{d_0}=\frac{\pi}{-ia/2}=\frac{2\pi}{a}\cdot\frac{1}{-i}
   =\frac{2\pi}{a}\,i ,
\]
a point on the positive imaginary axis, which does not lie on the cut
$(-\infty,0]$ of the principal square root; therefore
$z\mapsto z^{1/2}$ is continuous at $\pi/d_0$, and
$\varepsilon\mapsto\pi/d_\varepsilon$ is continuous at $\varepsilon=0$ because
$d_0\ne0$. Consequently
$\varepsilon\mapsto(\pi/d_\varepsilon)^{1/2}$ is continuous at
$\varepsilon=0$, with value $\bigl((2\pi/a)i\bigr)^{1/2}$. Its modulus is
$|(2\pi/a)i|^{1/2}=\sqrt{2\pi/a}$ and its argument is one half of
$\arg\bigl((2\pi/a)i\bigr)=\pi/2$, i.e. $\pi/4$; hence
\begin{equation}
  \label{eq:sqrtvalue}
  \Bigl(\frac{2\pi}{a}i\Bigr)^{1/2}=\sqrt{\frac{2\pi}{a}}\,e^{i\pi/4}
\end{equation}
(Lean: \LN{cpow_half_mul_I}, \LN{continuousAt_gaussian_value}), so that
$\text{(iii)}\to0$ as $\varepsilon\to0^{+}$.

\emph{(f) Conclusion.} For every $\varepsilon>0$ we have obtained
\[
  \Bigl\|\int_{-R}^{R}e^{iax^2/2}dx-\sqrt{\tfrac{2\pi}{a}}e^{i\pi/4}\Bigr\|
   \ \le\ 2R\bigl(1-e^{-\varepsilon R^2}\bigr)+\frac{4}{aR}
     +\Bigl\|\bigl(\pi/d_\varepsilon\bigr)^{1/2}
        -\sqrt{\tfrac{2\pi}{a}}e^{i\pi/4}\Bigr\| .
\]
The left-hand side does not depend on $\varepsilon$; by (c) and (e) the first
and the third summands on the right tend to zero as $\varepsilon\to0^{+}$,
while the middle one is constant. Letting $\varepsilon\to0^{+}$ therefore
gives
\[
  \Bigl\|\int_{-R}^{R}e^{iax^2/2}dx-\sqrt{\tfrac{2\pi}{a}}e^{i\pi/4}\Bigr\|
   \le\frac{4}{aR}\qquad(R>0),
\]
which is the quantitative statement; letting now $R\to\infty$ we get
$\int_{-R}^{R}e^{iax^2/2}dx\to\sqrt{2\pi/a}\,e^{i\pi/4}$. (The corresponding
statements have been machine-checked: \LN{fresnel_bound},
\LN{fresnel_one_dim}.)
\end{proof}

\begin{remark}
\label{rem:nowick}
The damping factor $e^{-\varepsilon x^2}$ is an \emph{Abelian regularisation}
in a real parameter, not a Wick rotation: time is not rotated into the complex
plane, the regularisation is removed inside the estimate itself, and the final
statement of Theorem~\ref{thm:fresnel1} concerns exclusively the real-time
oscillating integral. The key nontrivial point is the uniformity of the tail
estimate (Theorem~\ref{thm:tail}) in the damping parameter.
\end{remark}

\subsection{The oscillating Gaussian premeasure}

Let $A$ be a positive-definite symmetric operator, diagonal in the mode basis,
with eigenvalues $a_i>0$ (this is exactly how the free field looks in its own
modes: $A$ is the discrete Laplacian plus $m^2$). Define
\[
  \mu^{\mathrm{osc}}_V(B)=\int_B \exp\Bigl(\frac{i}{2}\ip{y}{Ay}\Bigr)dy,
  \qquad B\subseteq V ,
\]
(Lean: \LN{oscGaussianPremeasure}); the normalised premeasure is obtained by
dividing by the total Fresnel mass (Lean:
\LN{oscGaussianPremeasureNormalized}).

\begin{theorem}[Projective consistency;\\ Lean:
\LN{oscGaussian_projectivelyConsistent}]
\label{thm:consistency}
For every product of sets $\prod_{i<k}s_i$ and for every additional mode
with eigenvalue $a_k>0$,
\[
  \lim_{R\to\infty}\mu^{\mathrm{osc}}_{V_{k+1}}
    \Bigl(\prod_{i<k}s_i\times[-R,R]\Bigr)
  =\mu^{\mathrm{osc}}_{V_k}\Bigl(\prod_{i<k}s_i\Bigr)\cdot
    \sqrt{\frac{2\pi}{a_k}}\,e^{i\pi/4}.
\]
In other words, adding a mode multiplies the premeasure exactly by the Fresnel
mass of that mode, so the \emph{normalised} family satisfies the Kolmogorov
identity.
\end{theorem}

\begin{proof}
\emph{(a) Data.} Fix $k$ modes with eigenvalues $a_0,\dots,a_{k-1}>0$, an
additional mode with eigenvalue $a_k>0$, and measurable sets
$s_0,\dots,s_{k-1}\subseteq\RR$.

\emph{(b) Factorisation of the $(k+1)$-fold premeasure.} On
$\prod_{i<k}s_i\times[-R,R]$ the oscillating kernel factorises:
\[
  \exp\Bigl(\frac i2\sum_{i\le k}a_iy_i^2\Bigr)
   =\Bigl(\prod_{i<k}e^{ia_iy_i^2/2}\Bigr)\cdot e^{ia_ky_k^2/2}.
\]
The integrand is bounded (its modulus is $1$) and continuous, and the region
is a product of a set of finite Lebesgue measure with a bounded interval; so
Fubini's theorem for the product measure applies and
\begin{equation}
  \label{eq:premeasurefact}
  \mu^{\mathrm{osc}}_{V_{k+1}}\Bigl(\prod_{i<k}s_i\times[-R,R]\Bigr)
   =\mu^{\mathrm{osc}}_{V_k}\Bigl(\prod_{i<k}s_i\Bigr)\cdot
     \int_{-R}^{R}e^{ia_kx^2/2}\,dx .
\end{equation}

\emph{(c) The added factor converges.} By Theorem~\ref{thm:fresnel1} applied
with $a=a_k>0$,
\[
  \lim_{R\to\infty}\int_{-R}^{R}e^{ia_kx^2/2}\,dx
   =\sqrt{\frac{2\pi}{a_k}}\,e^{i\pi/4}.
\]

\emph{(d) Passage to the limit.} The first factor
in~\eqref{eq:premeasurefact} does not depend on $R$; multiplying a convergent
sequence by a constant,
\[
  \lim_{R\to\infty}\mu^{\mathrm{osc}}_{V_{k+1}}
   \Bigl(\prod_{i<k}s_i\times[-R,R]\Bigr)
  =\mu^{\mathrm{osc}}_{V_k}\Bigl(\prod_{i<k}s_i\Bigr)
   \cdot\sqrt{\frac{2\pi}{a_k}}\,e^{i\pi/4},
\]
which is the assertion.

\emph{(e) The normalised premeasure.} Write
$Z_n:=\prod_{i<n}\sqrt{2\pi/a_i}\,e^{i\pi/4}$ for the total Fresnel mass of
the first $n$ modes; each factor is nonzero, so $Z_n\ne0$, and
$Z_{k+1}=Z_k\cdot\sqrt{2\pi/a_k}\,e^{i\pi/4}$. Dividing the previous display
by $Z_{k+1}$ gives the Kolmogorov identity for the normalised premeasure
$\widehat\mu^{\mathrm{osc}}=\mu^{\mathrm{osc}}/Z$:
\begin{equation}
  \label{eq:kolmogorov}
  \lim_{R\to\infty}\widehat\mu^{\mathrm{osc}}_{V_{k+1}}
    \Bigl(\prod_{i<k}s_i\times[-R,R]\Bigr)
  =\widehat\mu^{\mathrm{osc}}_{V_k}\Bigl(\prod_{i<k}s_i\Bigr).
\end{equation}
(The corresponding statement has been machine-checked:
\LN{oscGaussian_projectivelyConsistent}.)
\end{proof}

\subsection{The free field}

The free-field action in the first $k$ modes is
$S(\varphi)=\tfrac12\ip{\varphi}{A\varphi}
 =\tfrac12\sum_i a_i\ip{e_i}{\varphi}^2$
(Lean: \LN{freeFieldAction}), and the Feynman functional is $e^{iS(\varphi)}$
(Lean: \LN{oscKernel_coords}).

\begin{theorem}[Existence of the HK integral; Lean: \LN{freeFieldHKIntegral},
\LN{freeFieldHKIntegral_cyl}]
\label{thm:exists}
For every bounded box $I\subseteq\RR^k$ the functional $e^{iS}$ is
Henstock--Kurzweil integrable on $I$ (for every set of integration parameters,
in particular for Henstock's gauge integral), and its HK integral equals the
integral over the box. Correspondingly, on the cylinder $\mathrm{cyl}_e(I)$
the functional $\varphi\mapsto e^{iS(\varphi)}$ is HK-integrable along
arbitrary cylindrical divisions with respect to the Gaussian cylindrical
premeasure.
\end{theorem}

\begin{proof}
The phase $y\mapsto\tfrac12\ip{y}{Ay}=\tfrac12\sum_ia_iy_i^2$ is a polynomial,
hence continuous on $\RR^k$; the complex exponential is continuous on $\CC$; so
the composite $y\mapsto e^{iS(y)}$ is continuous on $\RR^k$, in particular on
the compact closed box $\overline I$, where it is moreover bounded by $1$ in
modulus.

A function continuous on a compact box is Riemann integrable there, and for a
function which is Riemann integrable the Henstock--Kurzweil (gauge) integral
exists and has the same value: given $\varepsilon>0$, uniform continuity
furnishes a $\rho>0$ for which every $\rho$-fine tagged division has Riemann
sum within $\varepsilon$ of $\int_Ie^{iS}$, and the constant gauge
$\delta\equiv\rho$ then witnesses HK integrability. This argument does not use
any property of the set of integration parameters beyond the existence of
$\delta$-fine divisions (Cousin's lemma), so it applies to every such set, in
particular to Henstock's.

For the cylindrical statement, put $f(y)=e^{iS(y)}$ and
$F=f\circ\mathrm{coords}_e$. Since $f$ is continuous, Corollary~\ref{cor:arb}
applies verbatim and gives HK integrability of $F$ over $\mathrm{cyl}_e(I)$
along arbitrary cylindrical divisions, with value $\int_If\,d\gamma$. (The
corresponding statements have been machine-checked:
\LN{freeFieldHKIntegral}, \LN{freeFieldHKIntegral_cyl}.)
\end{proof}

\begin{theorem}[Fresnel formula for the free field; Lean:
\LN{fresnel_multi}, \LN{freeField_fresnel_formula}]
\label{thm:fresnelfree}
Let $a_i>0$ for $i<k$ and $A=\mathrm{diag}(a)$. Then
\[
  \lim_{R\to\infty}\int_{[-R,R]^k}
    \exp\Bigl(\frac{i}{2}\ip{y}{Ay}\Bigr)dy
  =(2\pi)^{k/2}\,e^{i\pi k/4}\,(\det A)^{-1/2}.
\]
\end{theorem}

\begin{proof}
\emph{(a) Factorisation of the integrand.} For $A=\mathrm{diag}(a)$ one has
$\ip{y}{Ay}=\sum_{i<k}a_iy_i^2$, so
\[
  \exp\Bigl(\frac i2\ip{y}{Ay}\Bigr)=\prod_{i<k}e^{ia_iy_i^2/2},
\]
a product of one-dimensional factors, the $i$-th depending on $y_i$ only.

\emph{(b) Factorisation of the integral.} The cube $[-R,R]^k$ is the product
of the intervals $[-R,R]$, and the Lebesgue measure of $\RR^k$ is the product
of the Lebesgue measures of the factors. The integrand is continuous and
bounded in modulus by $1$ on the compact cube, hence integrable there, so
Fubini's theorem applies and
\begin{equation}
  \label{eq:cubefact}
  \int_{[-R,R]^k}\exp\Bigl(\frac i2\ip{y}{Ay}\Bigr)dy
   =\prod_{i<k}\int_{-R}^{R}e^{ia_ix^2/2}\,dx .
\end{equation}

\emph{(c) The one-dimensional limits.} By Theorem~\ref{thm:fresnel1}, for each
$i<k$,
\[
  \lim_{R\to\infty}\int_{-R}^{R}e^{ia_ix^2/2}dx=\sqrt{\frac{2\pi}{a_i}}\,e^{i\pi/4}.
\]

\emph{(d) The product of the limits.} A finite product of convergent families
converges to the product of the limits, so, letting $R\to\infty$
in~\eqref{eq:cubefact},
\[
  \lim_{R\to\infty}\int_{[-R,R]^k}\exp\Bigl(\frac i2\ip{y}{Ay}\Bigr)dy
   =\prod_{i<k}\sqrt{\frac{2\pi}{a_i}}\,e^{i\pi/4}
   =(2\pi)^{k/2}\,e^{i\pi k/4}\Bigl(\prod_{i<k}a_i\Bigr)^{-1/2},
\]
the last equality collecting the $k$ factors $\sqrt{2\pi}$, the $k$ factors
$e^{i\pi/4}$ and the $k$ factors $a_i^{-1/2}$ separately.

\emph{(e) The determinant.} For a diagonal matrix the determinant is the
product of the diagonal entries, $\det A=\prod_{i<k}a_i$, which turns the
previous display into
$(2\pi)^{k/2}e^{i\pi k/4}(\det A)^{-1/2}$, as asserted. (The corresponding
statements have been machine-checked: \LN{fresnel_multi},
\LN{freeField_fresnel_formula}.)
\end{proof}

\begin{remark}
\label{rem:scope}
The formalised version covers an action diagonal in the chosen mode basis;
this is no loss of generality for the free field, since the mode basis is
exactly the eigenbasis of $A$. The infinite-dimensional extension is achieved
by the general Cousin lemma and the $K_0$-reduction of the previous stage:
every $\delta$-fine division of a cylinder is reduced to finitely many modes,
and consistency (Theorems~\ref{thm:vnvm} and \ref{thm:consistency}) guarantees
that the value does not depend on their number.
\end{remark}

\begin{figure}[ht]
\centering
\begin{tikzpicture}[scale=1.0]
  \draw[->] (0,0) -- (6.2,0) node[right] {$R$};
  \draw[->] (0,-1.6) -- (0,1.9) node[above] {$\int_{-R}^{R}e^{iax^2/2}dx$};
  \draw[dashed] (0,1.2) -- (6.0,1.2) node[right] {\small limit};
  \draw[thick,domain=0.15:6.0,samples=200,smooth]
    plot (\x,{1.2 - 0.9*cos(180*\x)/(1+1.6*\x)});
  \draw[thick,dotted,domain=0.6:6.0,samples=100,smooth]
    plot (\x,{1.2 + 0.9/(1+1.6*\x)});
  \draw[thick,dotted,domain=0.6:6.0,samples=100,smooth]
    plot (\x,{1.2 - 0.9/(1+1.6*\x)});
\end{tikzpicture}
\caption{Oscillating convergence of the truncated Fresnel integrals and the
``funnel'' $4/(aR)$ of Theorem~\ref{thm:fresnel1}: the convergence is genuine
but not absolute.}
\label{fig:fresnel}
\end{figure}
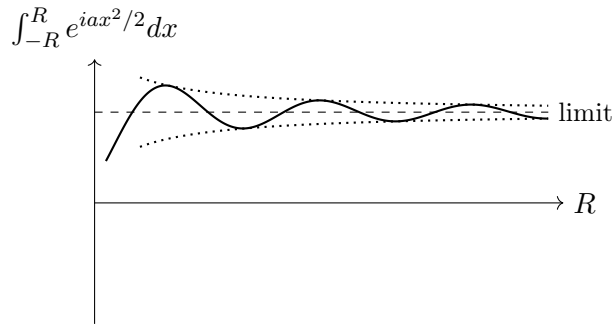

\section{The non-diagonal free field and the generating functional}
\label{sec:general}

Sections~\ref{sec:general}--\ref{sec:cylgeneral} and \ref{sec:limit} bring the
description of the free field to a complete closure: the assumption that the action is diagonal
(Remark~\ref{rem:scope}) is removed, a generating functional with an external
source is introduced, the two-point function is computed, all $n$-point
functions are obtained from Wick's theorem (Section~\ref{sec:wick}), the
oscillating
premeasure is embedded into the general theory of cylindrical premeasures, the
value of the HK integral is identified with the Fresnel value, and finally the
passage from the damped (Euclidean) picture to the oscillating (real-time) one
is established. The interacting field, unbounded generators and
renormalisation remain open (see Section~\ref{sec:concl}).

\subsection{Orthogonal diagonalisation of a positive-definite action}
\label{ssec:diag}

Let $A$ be a real symmetric positive-definite $k\times k$ matrix. By the
spectral theorem there exist an orthogonal matrix $U$ (in Lean,
\LN{Matrix.IsHermitian.eigenvectorUnitary}) and eigenvalues $a_i>0$ such that
\begin{equation}
  \label{eq:spectral}
  A=U\,\mathrm{diag}(a)\,U^{\top},\qquad
  \ip{y}{Ay}=\sum_{i<k}a_i\,(O y)_i^2,\qquad O:=U^{\top},
\end{equation}
with $|\det O|=1$ and $\det A=\prod_{i<k}a_i$. This is formalised as
\LN{InfiniteHK3.eigenRotation_spec}; the matrix $O$ and the family $a$
themselves are introduced as \LN{eigenRotation} and \LN{eigenModes}.

\begin{remark}[On the choice of the exhausting family]
\label{rem:exhaustion}
The integral $\int e^{(i/2)\ip{y}{Ay}}dy$ converges only conditionally, so its
value depends on the family of bounded regions used to exhaust $\RR^k$. The
orthogonal change of variables $y=Uz$ preserves Lebesgue measure but maps the
cube $[-R,R]^k$ onto a rotated cube. Therefore the natural (and the only
correct) exhausting family for a non-diagonal action is the family of cubes
adapted to the eigenbasis:
\begin{equation}
  \label{eq:adapted}
  Q_R(A,c):=\{\,y\in\RR^k:\ O(y+c)\in[-R,R]^k\,\},
\end{equation}
in Lean \LN{adaptedCube}. For a diagonal $A$ and $c=0$ these are exactly the
cubes $[-R,R]^k$, so that Theorem~\ref{thm:fresnelfree} is the special case of
Theorem~\ref{thm:fresnelgeneral} below.
\end{remark}

\subsection{The Fresnel formula for a general action}
\label{ssec:fresnelgeneral}

\begin{theorem}[Fresnel formula for a general positive-definite action; Lean:
\LN{freeField_fresnel_formula_general}]
\label{thm:fresnelgeneral}
Let $A$ be a real symmetric positive-definite $k\times k$ matrix. Then
\[
  \lim_{R\to\infty}\int_{Q_R(A,0)}
    \exp\Bigl(\frac{i}{2}\ip{y}{Ay}\Bigr)dy
  =(2\pi)^{k/2}\,e^{i\pi k/4}\,(\det A)^{-1/2},
\]
where $Q_R(A,0)$ is the adapted cube~\eqref{eq:adapted}. The right-hand side is
introduced in Lean as \LN{fresnelValueGeneral}.
\end{theorem}

\begin{proof}
\emph{(a) Spectral decomposition.} The matrix $A$ is real symmetric, so by the
spectral theorem there are an orthogonal matrix $U$ ($U^{\top}U=1$) and real
numbers $a_0,\dots,a_{k-1}$ with $A=U\,\mathrm{diag}(a)\,U^{\top}$; writing
$O:=U^{\top}$, this reads $A=O^{\top}\mathrm{diag}(a)\,O$. Positive
definiteness of $A$ gives $a_i=\ip{Ue_i}{A\,Ue_i}>0$ for each $i$ (Lean:
\LN{eigenRotation_spec}).

\emph{(b) The quadratic form in the new variable.} Put $z=Oy$. Then
\[
  \ip{y}{Ay}=\ip{y}{O^{\top}\mathrm{diag}(a)Oy}
   =\ip{Oy}{\mathrm{diag}(a)\,Oy}=\ip{z}{\Lambda z}=\sum_{i<k}a_iz_i^2 ,
  \qquad \Lambda:=\mathrm{diag}(a),
\]
which is~\eqref{eq:spectral}.

\emph{(c) The change of variable preserves Lebesgue measure.} The matrix $O$
is orthogonal, hence $y\mapsto Oy$ is a linear isometry of $\RR^k$ onto itself
and $(\det O)^2=\det(O^{\top}O)=1$, so $|\det O|=1$. By the transformation
formula for Lebesgue measure under a linear map, $y\mapsto Oy$ is
measure-preserving (Lean: \LN{measurePreserving_mulVec_add} together with
\LN{Real.map_matrix_volume_pi_eq_smul_volume_pi}); being a linear bijection,
it is also a measurable embedding
(\LN{measurableEmbedding_mulVec_add}).

\emph{(d) The exhausting family.} By the definition~\eqref{eq:adapted} of the
adapted cube,
\[
  Q_R(A,0)=\{y:\ Oy\in[-R,R]^k\}=(y\mapsto Oy)^{-1}\bigl([-R,R]^k\bigr),
\]
so the image of $Q_R(A,0)$ under $y\mapsto Oy$ is exactly the cube
$[-R,R]^k$.

\emph{(e) Change of variables.} Combining (b), (c) and (d), the
change-of-variables formula for a measure-preserving measurable embedding
(Lean: \LN{setIntegral_adaptedCube}) gives, for every $R>0$,
\begin{equation}
  \label{eq:cov-adapted}
  \int_{Q_R(A,0)}\exp\Bigl(\frac i2\ip{y}{Ay}\Bigr)dy
   =\int_{[-R,R]^k}\exp\Bigl(\frac i2\ip{z}{\Lambda z}\Bigr)dz .
\end{equation}

\emph{(f) Evaluation.} The right-hand side of~\eqref{eq:cov-adapted} is the
integral treated in Theorem~\ref{thm:fresnelfree} for the diagonal action
$\Lambda$, whose limit as $R\to\infty$ is
$(2\pi)^{k/2}e^{i\pi k/4}(\det\Lambda)^{-1/2}$. Finally
\[
  \det\Lambda=\prod_{i<k}a_i=\det\bigl(O^{\top}\Lambda O\bigr)=\det A ,
\]
since $\det(O^{\top})\det(O)=1$. Substituting into~\eqref{eq:cov-adapted} and
letting $R\to\infty$ completes the proof.

\emph{(g) A remark on the exhausting family.} The integral converges only
conditionally, so the choice of the exhausting family is not a matter of
convenience but part of the statement (Remark~\ref{rem:exhaustion}); the
family $Q_R(A,0)$ is the one for which the change of variables of step (e)
leads back to a cube in the new coordinates. For a diagonal $A$ one may take
$O=1$, and the adapted cube is literally $[-R,R]^k$, so that
Theorem~\ref{thm:fresnelfree} is the special case of the present theorem.
(The corresponding statement has been machine-checked:
\LN{freeField_fresnel_formula_general}.)
\end{proof}

\subsection{The generating functional}
\label{ssec:genfunc}

Introduce the action with an external source $j\in\RR^k$,
\[
  S_j(y)=\tfrac12\ip{y}{Ay}+\ip{j}{y}
\]
(in Lean \LN{sourceAction}, with kernel $e^{iS_j}$ given by
\LN{sourceKernel}), and set
\begin{equation}
  \label{eq:genfunc}
  Z(j):=\lim_{R\to\infty}\int_{Q_R(A,A^{-1}j)}e^{iS_j(y)}\,dy
\end{equation}
(in Lean, \LN{freeFieldGeneratingFunctional}), where the exhausting family is
shifted by $-A^{-1}j$, i.e.\ centred at the stationary point of the phase,
which is the correct analogue of completing the square.

\begin{theorem}[Value of the generating functional]
\label{thm:genfunc}
(Lean: \LN{freeFieldGeneratingFunctional_value},
\LN{freeFieldGeneratingFunctional_tendsto}.)
For a symmetric positive-definite $A$ and every $j\in\RR^k$ the limit
in~\eqref{eq:genfunc} exists and
\[
  Z(j)=(2\pi)^{k/2}\,e^{i\pi k/4}\,(\det A)^{-1/2}\,
        \exp\Bigl(-\frac{i}{2}\ip{j}{A^{-1}j}\Bigr).
\]
In particular $Z(0)=(2\pi)^{k/2}e^{i\pi k/4}(\det A)^{-1/2}\ne 0$
(Lean: \LN{freeFieldGeneratingFunctional_zero},
\LN{fresnelValueGeneral_ne_zero}).
\end{theorem}

\begin{proof}
\emph{(a) Completing the square.} Write $u=y+A^{-1}j$ and expand the quadratic
form, using bilinearity:
\[
  \ip{u}{Au}=\ip{y}{Ay}+\ip{y}{AA^{-1}j}+\ip{A^{-1}j}{Ay}
              +\ip{A^{-1}j}{AA^{-1}j}.
\]
Here $AA^{-1}=1$, so the second term is $\ip{y}{j}$ and the fourth is
$\ip{A^{-1}j}{j}=\ip{j}{A^{-1}j}$. For the third, the symmetry of $A$ gives
$\ip{A^{-1}j}{Ay}=\ip{A\,A^{-1}j}{y}=\ip{j}{y}$. The two middle terms are
therefore equal and add up to $2\ip{j}{y}$ --- not $\ip{j}{y}$ --- so that
\begin{equation}
  \label{eq:square}
  \ip{y+A^{-1}j}{A(y+A^{-1}j)}=\ip{y}{Ay}+2\ip{j}{y}+\ip{j}{A^{-1}j}
\end{equation}
(Lean: \LN{quadAction_shift}). Dividing by $2$,
$\tfrac12\ip{u}{Au}=S_j(y)+\tfrac12\ip{j}{A^{-1}j}$.

\emph{(b) Factorisation of the kernel.} Multiplying the last identity by $i$
and exponentiating,
\begin{equation}
  \label{eq:kernelfact}
  e^{iS_j(y)}=e^{-\frac i2\ip{j}{A^{-1}j}}\;e^{\frac i2\ip{u}{Au}},
  \qquad u=y+A^{-1}j .
\end{equation}

\emph{(c) The shifted exhausting family.} By definition
$Q_R(A,A^{-1}j)=\{y:\ O(y+A^{-1}j)\in[-R,R]^k\}$, so the translation
$y\mapsto u=y+A^{-1}j$ maps $Q_R(A,A^{-1}j)$ bijectively onto $Q_R(A,0)$. A
translation preserves Lebesgue measure, hence
\[
  \int_{Q_R(A,A^{-1}j)}e^{\frac i2\ip{y+A^{-1}j}{A(y+A^{-1}j)}}dy
   =\int_{Q_R(A,0)}e^{\frac i2\ip{u}{Au}}\,du .
\]

\emph{(d) Conclusion.} The prefactor in~\eqref{eq:kernelfact} is a constant,
independent of $y$, so it may be taken outside the integral:
\[
  \int_{Q_R(A,A^{-1}j)}e^{iS_j(y)}dy
   =e^{-\frac i2\ip{j}{A^{-1}j}}\int_{Q_R(A,0)}e^{\frac i2\ip{u}{Au}}du .
\]
By Theorem~\ref{thm:fresnelgeneral} the integral on the right converges, as
$R\to\infty$, to $(2\pi)^{k/2}e^{i\pi k/4}(\det A)^{-1/2}$; multiplying a
convergent family by a fixed constant, the limit~\eqref{eq:genfunc} exists and
\[
  Z(j)=(2\pi)^{k/2}e^{i\pi k/4}(\det A)^{-1/2}
        \exp\Bigl(-\frac i2\ip{j}{A^{-1}j}\Bigr).
\]

\emph{(e) Non-vanishing at $j=0$.} At $j=0$ the exponential factor is $1$ and
$Z(0)=(2\pi)^{k/2}e^{i\pi k/4}(\det A)^{-1/2}$. Here $(2\pi)^{k/2}>0$ and
$(\det A)^{-1/2}>0$, because $\det A=\prod_ia_i>0$ for a positive-definite
$A$, while $|e^{i\pi k/4}|=1$; a product of a strictly positive real number
and a complex number of modulus one is nonzero, so $Z(0)\ne0$ (Lean:
\LN{freeFieldGeneratingFunctional_zero}, \LN{fresnelValueGeneral_ne_zero}).
This is what makes the normalisation $\widehat Z=Z/Z(0)$ legitimate. (The
corresponding statements have been machine-checked:
\LN{freeFieldGeneratingFunctional_value},
\LN{freeFieldGeneratingFunctional_tendsto}.)
\end{proof}

\subsection{The two-point function}
\label{ssec:twopoint}

The normalised generating functional is
$\widehat Z(j)=Z(j)/Z(0)=\exp\bigl(-\tfrac{i}{2}\ip{j}{A^{-1}j}\bigr)$
(Lean: \LN{normalizedGeneratingFunctional},
\LN{normalizedGeneratingFunctional_eq}). Differentiating it with respect to the
source yields the correlations. We use the second route allowed by the
formulation of the problem: the two-point function is defined as the second
derivative of $\widehat Z$ at zero source, computed through a polarisation
identity, rather than as the ratio of two conditionally convergent integrals
(the latter would require a separate convergence analysis of
$\int y_py_q e^{iS}$).

Namely, for $v\in\RR^k$
\begin{equation}
  \label{eq:secondderiv}
  \frac{d^2}{dt^2}\Big|_{t=0}\widehat Z(tv)=-\,i\,\ip{v}{A^{-1}v}
\end{equation}
(in Lean, \LN{deriv_two_normalizedGeneratingFunctional}), and, writing $e_p$
for the standard unit vectors, we define
\[
  \langle y_py_q\rangle:=
  -\frac12\Bigl[
    \frac{d^2}{dt^2}\Big|_{0}\widehat Z\bigl(t(e_p+e_q)\bigr)
   -\frac{d^2}{dt^2}\Big|_{0}\widehat Z(te_p)
   -\frac{d^2}{dt^2}\Big|_{0}\widehat Z(te_q)\Bigr]
\]
(Lean: \LN{freeFieldTwoPointHK}).

\begin{theorem}[Two-point function; Lean: \LN{freeField_two_point_HK}]
\label{thm:twopoint}
For a symmetric positive-definite $A$ and all $p,q<k$,
\[
  \langle y_py_q\rangle=i\,(A^{-1})_{pq}.
\]
The quantity $(A^{-1})_{pq}$ is introduced as \LN{freeFieldTwoPoint}
(\LN{freeFieldTwoPoint_eq_inv}).
\end{theorem}

\begin{proof}
\emph{(a) Restriction to a line.} Let $v\in\RR^k$ and $t\in\RR$. By bilinearity
of the form $\ip{\cdot}{A^{-1}\cdot}$,
\[
  \ip{tv}{A^{-1}(tv)}=t^2\ip{v}{A^{-1}v},
\]
so that, with $c:=-\tfrac i2\ip{v}{A^{-1}v}$,
\[
  \widehat Z(tv)=\exp\Bigl(-\frac i2\,t^2\ip{v}{A^{-1}v}\Bigr)=\exp\bigl(ct^2\bigr).
\]

\emph{(b) Two derivatives at the origin.} The function $g(t)=\exp(ct^2)$ is
smooth, with $g'(t)=2ct\exp(ct^2)$ and
$g''(t)=\bigl(2c+4c^2t^2\bigr)\exp(ct^2)$; at $t=0$ this gives
$g''(0)=2c$. Hence for every $v\in\RR^k$
\begin{equation}
  \label{eq:secondderiv-proof}
  \frac{d^2}{dt^2}\Big|_{t=0}\widehat Z(tv)=2c=-\,i\,\ip{v}{A^{-1}v},
\end{equation}
which is~\eqref{eq:secondderiv} (Lean:
\LN{deriv_two_normalizedGeneratingFunctional}).

\emph{(c) Polarisation.} The matrix $A$ is symmetric, hence so is $A^{-1}$,
and the form $\ip{\cdot}{A^{-1}\cdot}$ is a symmetric bilinear form.
Expanding it on $e_p+e_q$,
\[
  \ip{e_p+e_q}{A^{-1}(e_p+e_q)}
   =\ip{e_p}{A^{-1}e_p}+\ip{e_p}{A^{-1}e_q}+\ip{e_q}{A^{-1}e_p}
     +\ip{e_q}{A^{-1}e_q},
\]
and $\ip{e_p}{A^{-1}e_q}=(A^{-1})_{pq}$,
$\ip{e_q}{A^{-1}e_p}=(A^{-1})_{qp}=(A^{-1})_{pq}$. Therefore
\begin{equation}
  \label{eq:polarisation}
  \ip{e_p+e_q}{A^{-1}(e_p+e_q)}-\ip{e_p}{A^{-1}e_p}-\ip{e_q}{A^{-1}e_q}
   =2\,(A^{-1})_{pq}.
\end{equation}

\emph{(d) Conclusion.} Substituting~\eqref{eq:secondderiv-proof} with
$v=e_p+e_q$, $v=e_p$ and $v=e_q$ into the definition of
$\langle y_py_q\rangle$ and using~\eqref{eq:polarisation},
\[
  \langle y_py_q\rangle
   =-\frac12\Bigl[-i\ip{e_p+e_q}{A^{-1}(e_p+e_q)}+i\ip{e_p}{A^{-1}e_p}
      +i\ip{e_q}{A^{-1}e_q}\Bigr]
   =\frac i2\cdot2\,(A^{-1})_{pq}=i\,(A^{-1})_{pq}.
\]

\emph{(e) Why this route.} We do not define $\langle y_py_q\rangle$ as the
ratio $\int y_py_q e^{iS}\big/\int e^{iS}$ of two conditionally convergent
integrals: the numerator has an unbounded polynomial factor, so its
convergence along the adapted exhausting family would require a separate
analysis, and its value would again depend on the exhaustion. Differentiating
the generating functional, whose value is already known in closed form by
Theorem~\ref{thm:genfunc}, avoids this entirely. (The corresponding statement
has been machine-checked: \LN{freeField_two_point_HK}.)
\end{proof}

\begin{remark}[The factor $i$]
\label{rem:factor-i}
In real time the ratio $\int y_py_q\,e^{iS}\,/\int e^{iS}$ equals precisely
$i(A^{-1})_{pq}$, not $(A^{-1})_{pq}$: the latter is the Euclidean (Gaussian)
answer. Check for $k=1$: analytic continuation of
$\int x^2e^{-bx^2/2}dx\big/\int e^{-bx^2/2}dx=1/b$ to $b=-ia$ gives
$1/(-ia)=i/a$. The factor $i$ is the usual factor of the Feynman propagator;
here it is obtained as a theorem, not postulated.
\end{remark}

\section{Wick's theorem and the $n$-point functions}
\label{sec:wick}

The generating functional of Theorem~\ref{thm:genfunc} determines not only the
two-point function of Theorem~\ref{thm:twopoint}, but every correlation
function of the free field. The present section closes the free field at the
level of \emph{all} $n$-point functions: they vanish in odd order and, in even
order, factorise into a sum over pairings of the two-point function. No new
analytic input is required: the whole argument consists in differentiating the
explicit exponential
$\widehat Z(j)=\exp\bigl(-\tfrac{i}{2}\ip{j}{A^{-1}j}\bigr)$ obtained in
\S\ref{ssec:genfunc}, and it is carried out in Lean under no hypothesis beyond
the positive definiteness of $A$.

\subsection{Definition of the $n$-point functions}
\label{ssec:npointdef}

The normalised generating functional $\widehat Z$ is a smooth function on
$\RR^k$ (Lean: \LN{contDiff_normalizedGeneratingFunctional}). We write
$\partial_pf(j)=Df(j)[e_p]$ for the partial derivative in the $p$-th
coordinate direction (Lean: \LN{partialDerivC}) and, for a list
$P=(p_1,\dots,p_n)$ of indices, $\partial_P=\partial_{p_1}\cdots\partial_{p_n}$
(Lean: \LN{iterPartialDerivC}). The order of the differentiations is
irrelevant: $\widehat Z$ is a smooth function on the finite-dimensional space
$\RR^k$, so all its iterated derivatives are continuous and Schwarz's theorem
on the symmetry of mixed partial derivatives applies; in Lean this is the
symmetry of the second derivative of a $C^2$ map together with
\LN{ContDiff.fderiv_apply}, which supplies the required smoothness of each
intermediate derivative.

\begin{definition}[$n$-point function; Lean: \LN{freeFieldNPoint}]
\label{def:npoint}
For a symmetric positive-definite $A$ and indices $p_1,\dots,p_n<k$,
\begin{equation}
  \label{eq:npointdef}
  \langle y_{p_1}\cdots y_{p_n}\rangle
   :=(-i)^n\,
     \frac{\partial^n}{\partial j_{p_1}\cdots\partial j_{p_n}}
     \widehat Z(j)\Big|_{j=0},
  \qquad
  \widehat Z(j)=\exp\Bigl(-\frac{i}{2}\ip{j}{A^{-1}j}\Bigr).
\end{equation}
\end{definition}

For $n=2$ this agrees with the two-point function of \S\ref{ssec:twopoint},
which was defined there through a polarisation identity (Lean:
\LN{freeFieldNPoint_two_eq_twoPointHK}).

The right-hand side of Wick's theorem is a sum over the perfect matchings
(pairings) of the list of indices. We encode it by the standard recursion: the
first index is paired with each of the remaining ones in turn, and both paired
indices are deleted. Denote by $\mathrm{holes}(P)$ the list of pairs $(q,P')$
in which $q$ runs over the entries of $P$ and $P'$ is $P$ with that entry
deleted (Lean: \LN{holes}).

\begin{definition}[Pairing sum; Lean: \LN{wickPairingSum}]
\label{def:pairingsum}
For a weight $B$ put $W_B(\varnothing)=1$ and
\[
  W_B(p\!::\!P)=\sum_{(q,P')\in\mathrm{holes}(P)}B_{pq}\,W_B(P').
\]
\end{definition}

If $P$ has $2m$ entries, then $W_B(P)$ is precisely
$\sum_{\pi}\prod_{(i,j)\in\pi}B_{p_ip_j}$, the sum over the $(2m-1)!!$ perfect
matchings $\pi$ of the positions $\{1,\dots,2m\}$, each unordered pair being
counted once; if $P$ has odd length, the recursion terminates in an empty sum,
so that $W_B(P)=0$.

\subsection{Wick's theorem: the odd case}
\label{ssec:wickodd}

\begin{theorem}[Odd correlations vanish; Lean: \LN{freeFieldNPoint_odd}]
\label{thm:wickodd}
Let $A$ be symmetric positive-definite. Then for every list of odd length,
\[
  \langle y_{p_1}\cdots y_{p_{2m+1}}\rangle=0 .
\]
\end{theorem}

\begin{proof}
By Theorem~\ref{thm:wick} below, the correlation equals the pairing sum
$W_B(P)$ with $B_{pq}=i(A^{-1})_{pq}$, so it suffices to show that
$W_B(P)=0$ for every list $P$ of odd length. We prove this by induction on
the odd length $2m+1$.

For $m=0$ the list has one entry, $P=(p)$, and
Definition~\ref{def:pairingsum} gives
$W_B\bigl(p::\varnothing\bigr)=\sum_{(q,P')\in\mathrm{holes}(\varnothing)}
B_{pq}W_B(P')$; the list $\mathrm{holes}(\varnothing)$ is empty, since there
is no entry to delete, so the sum is empty and $W_B(P)=0$.

For the step, let $P=p::P_0$ have length $2m+3$, so that $P_0$ has length
$2m+2$. Every pair $(q,P')\in\mathrm{holes}(P_0)$ has $|P'|=|P_0|-1=2m+1$,
again odd; by the induction hypothesis $W_B(P')=0$ for each of them, whence
\[
  W_B(P)=\sum_{(q,P')\in\mathrm{holes}(P_0)}B_{pq}\,W_B(P')=0 .
\]
This proves the claim for all odd lengths (Lean: \LN{wickPairingSum_odd}).
Equivalently, and more directly: $\widehat Z(-j)=\widehat Z(j)$, because the
quadratic form $\ip{j}{A^{-1}j}$ is even, so $\widehat Z$ is an even function
of the source and all its odd-order derivatives vanish at $j=0$. (The
corresponding statement has been machine-checked: \LN{freeFieldNPoint_odd}.)
\end{proof}

\subsection{Wick's theorem: the even case}
\label{ssec:wickeven}

Two elementary lemmas carry the whole proof.

\begin{lemma}[First derivative; Lean:
\LN{partialDerivC}\allowbreak\LN{_normalizedGeneratingFunctional}]
\label{lem:firstderiv}
For a symmetric positive-definite $A$,
\[
  \partial_p\widehat Z(j)=-\,i\,(A^{-1}j)_p\,\widehat Z(j).
\]
\end{lemma}

\begin{proof}
Along the line $t\mapsto j+te_p$ the quadratic form expands as
$\ip{j+te_p}{A^{-1}(j+te_p)}
 =\ip{j}{A^{-1}j}+2t\,(A^{-1}j)_p+t^2(A^{-1})_{pp}$,
where the symmetry of $A^{-1}$ is used (Lean: \LN{quadForm_line}).
Differentiating $\exp\bigl(-\tfrac{i}{2}(\,\cdot\,)\bigr)$ at $t=0$ gives the
stated factor $-i(A^{-1}j)_p$.
\end{proof}

\begin{lemma}[Leibniz rule for a linear factor; Lean:
\LN{iterPartialDerivC_linear_mul}]
\label{lem:leibniz}
Let $\ell(j)=\sum_qw_qj_q$ be a linear form with complex coefficients and let
$f$ be smooth. Then, for every list $P$,
\[
  \partial_P(\ell f)(j)=\ell(j)\,\partial_Pf(j)
   +\sum_{(q,P')\in\mathrm{holes}(P)}w_q\,\partial_{P'}f(j).
\]
\end{lemma}

\begin{proof}
Induction on $P$. For the empty list the sum is empty and the assertion is
trivial. For $p\!::\!P$ the product rule gives
$\partial_p(\ell f)=w_pf+\ell\,\partial_pf$; applying $\partial_P$, using its
additivity and homogeneity on smooth functions (Lean:
\LN{iterPartialDerivC_add}, \LN{iterPartialDerivC_const_mul}) and the
induction hypothesis for $\partial_pf$, one obtains exactly the required
right-hand side, the term $w_p\,\partial_Pf$ being the entry of
$\mathrm{holes}(p\!::\!P)$ that removes the head.
\end{proof}

\begin{theorem}[Wick's theorem; Lean: \LN{freeField_wick_theorem}]
\label{thm:wick}
Let $A$ be symmetric positive-definite and $P=(p_1,\dots,p_n)$. Then
\[
  \langle y_{p_1}\cdots y_{p_n}\rangle=W_B(P),
  \qquad B_{pq}=i\,(A^{-1})_{pq}=\langle y_py_q\rangle .
\]
Explicitly, for $n=2m$,
\begin{equation}
  \label{eq:wick}
  \langle y_{p_1}\cdots y_{p_{2m}}\rangle
   =\sum_{\pi\in\mathrm{Pairings}(2m)}\ \prod_{(i,j)\in\pi}
      \langle y_{p_i}y_{p_j}\rangle
   =i^{\,m}\sum_{\pi\in\mathrm{Pairings}(2m)}\ \prod_{(i,j)\in\pi}
      (A^{-1})_{p_ip_j},
\end{equation}
while all correlations of odd order vanish.
\end{theorem}

\begin{proof}
Write $D(P)=\partial_P\widehat Z(0)$. By Lemma~\ref{lem:firstderiv},
$\partial_{p}\widehat Z=\ell\,\widehat Z$ with the linear form
$\ell(j)=-i\sum_q(A^{-1})_{pq}j_q$, whose coefficients are
$w_q=-i(A^{-1})_{pq}$. Applying Lemma~\ref{lem:leibniz} with $f=\widehat Z$
and evaluating at $j=0$, where $\ell(0)=0$, we obtain the Wick recursion
\begin{equation}
  \label{eq:wickrec}
  D(p\!::\!P)=\sum_{(q,P')\in\mathrm{holes}(P)}
    \bigl(-i\,(A^{-1})_{pq}\bigr)\,D(P')
\end{equation}
(Lean: \LN{iterPartialDerivC_ngf_cons_zero}).

We now induct on the length $n$ of the list. For $n=0$ both sides equal $1$,
because $\widehat Z(0)=1$. Let $n\ge1$ and $P=p\!::\!P_0$. Every pair
$(q,P')\in\mathrm{holes}(P_0)$ satisfies $|P'|=n-2$. Multiplying
\eqref{eq:wickrec} by $(-i)^n$ and using
$(-i)^{|P'|+3}=(-i)^{|P'|}(-i)^3=i\,(-i)^{|P'|}$, we get
\[
  \langle y_{p}y_{P_0}\rangle
   =\sum_{(q,P')\in\mathrm{holes}(P_0)}
      i\,(A^{-1})_{pq}\cdot(-i)^{|P'|}D(P'),
\]
where the reindexing of the power of $-i$ is carried out as follows: the list
$p\!::\!P_0$ has $n=|P_0|+1$ entries and $|P'|=|P_0|-1=n-2$, so
$n=|P'|+2$ and
\[
  (-i)^n\cdot\bigl(-i(A^{-1})_{pq}\bigr)
   =(-i)^{|P'|+2}(-i)\,(A^{-1})_{pq}
   =(-i)^{|P'|}\,(-i)^3\,(A^{-1})_{pq}
   =(-i)^{|P'|}\cdot i\,(A^{-1})_{pq},
\]
because $(-i)^2=-1$ and hence $(-i)^3=(-i)^2(-i)=(-1)(-i)=i$.
The induction hypothesis identifies
$(-i)^{|P'|}D(P')$ with $W_B(P')$, and the coefficient $i(A^{-1})_{pq}$ that
has just been isolated is precisely the weight $B_{pq}$. This is exactly the
recursion of Definition~\ref{def:pairingsum} for $W_B$ with $B=i\,A^{-1}$.

Finally, $B_{pq}=i(A^{-1})_{pq}=\langle y_py_q\rangle$ by
Theorem~\ref{thm:twopoint}, which gives the first equality
of~\eqref{eq:wick}. For the second, let $n=2m$; a perfect matching $\pi$ of
$2m$ positions consists of exactly $m$ unordered pairs, and each pair
contributes one factor $B_{p_ip_j}=i(A^{-1})_{p_ip_j}$. Extracting the factor
$i$ from each of the $m$ pairs,
\[
  \prod_{(i,j)\in\pi}B_{p_ip_j}
   =\prod_{(i,j)\in\pi}i\,(A^{-1})_{p_ip_j}
   =i^{\,m}\prod_{(i,j)\in\pi}(A^{-1})_{p_ip_j},
\]
and summing over $\pi$ gives the closed form
$i^{\,m}\sum_\pi\prod_{(i,j)\in\pi}(A^{-1})_{p_ip_j}$ of~\eqref{eq:wick}.
(The corresponding statement has been machine-checked:
\LN{freeField_wick_theorem}.)
\end{proof}

\subsection{The four-point function and interpretation}
\label{ssec:fourpoint}

\begin{corollary}[Four-point function; Lean: \LN{freeFieldNPoint_four}]
\label{cor:fourpoint}
For a symmetric positive-definite $A$ and $p,q,r,s<k$,
\[
  \langle y_py_qy_ry_s\rangle
  =-\bigl[(A^{-1})_{pq}(A^{-1})_{rs}
         +(A^{-1})_{pr}(A^{-1})_{qs}
         +(A^{-1})_{ps}(A^{-1})_{qr}\bigr].
\]
\end{corollary}

\begin{proof}
Apply Theorem~\ref{thm:wick} with $n=2m=4$, so $m=2$. The perfect matchings of
the four positions $\{1,2,3,4\}$, labelled here by the indices $p,q,r,s$, are
exactly three:
\[
  pq\,|\,rs,\qquad pr\,|\,qs,\qquad ps\,|\,qr ,
\]
since the position of $p$ may be matched with any one of the three remaining
positions and the two positions then left over are forced to be matched with
each other. Formula~\eqref{eq:wick} therefore reads
\[
  \langle y_py_qy_ry_s\rangle
   =i^{\,2}\bigl[(A^{-1})_{pq}(A^{-1})_{rs}+(A^{-1})_{pr}(A^{-1})_{qs}
     +(A^{-1})_{ps}(A^{-1})_{qr}\bigr],
\]
and $i^2=-1$ gives the stated sign. (The corresponding statement has been
machine-checked: \LN{freeFieldNPoint_four}.)
\end{proof}

\begin{remark}[The real-time factor]
\label{rem:wickfactor}
The sign in Corollary~\ref{cor:fourpoint} is the real-time counterpart of the
Euclidean formula
$\langle y_py_qy_ry_s\rangle_{\mathrm{E}}
 =(A^{-1})_{pq}(A^{-1})_{rs}+(A^{-1})_{pr}(A^{-1})_{qs}
  +(A^{-1})_{ps}(A^{-1})_{qr}$:
each contraction carries the Feynman factor $i$ of
Remark~\ref{rem:factor-i}, so that a correlation of order $2m$ carries the
factor $i^{\,m}$. As in the case of the two-point function, this factor is
obtained here as a theorem, not postulated.
\end{remark}

\begin{remark}[What this closes]
\label{rem:wickclosure}
Theorem~\ref{thm:wick} completes the description of the free field in the
Henstock--Kurzweil (gauge) picture: the value of the functional integral
(Theorem~\ref{thm:fresnelgeneral}), the generating functional
(Theorem~\ref{thm:genfunc}) and now the whole hierarchy of correlation
functions are available, and all of them are consequences of the Fresnel
formula proved without any Wick rotation. Note that no additional analytic
work was needed: the derivatives are taken of the closed-form expression for
$\widehat Z$, which was itself obtained as a limit of Henstock--Kurzweil
integrals.
\end{remark}

\section{The oscillating premeasure in the general theory}
\label{sec:cylgeneral}

\subsection{The oscillating Gaussian cylindrical premeasure}
\label{ssec:cylosc}

The structure \LN{InfiniteHK.CylindricalPremeasure} of
stage~\cite{BerdinskyInfiniteHK} stores a family of \emph{measures} $\mu_V$
with values in $[0,\infty]$. The oscillating premeasure takes complex values
and has infinite total variation (its modulus is Lebesgue measure), so it
cannot in principle be an instance of that structure. We therefore introduce
its complex analogue \LN{InfiniteHK3.ComplexCylindricalPremeasure}: an
orthonormal mode system $e:\NN\to\HH$ together with a function
$(n,s)\mapsto\mu(n,s)\in\CC$ giving the value on the cylinder over the product
$\prod_{i<n}s_i$; the cylinder in $\HH$ itself is denoted \LN{cylinder}.

The oscillating Gaussian cylindrical premeasure and its normalised version are
\[
  \mu^{\mathrm{osc}}(n,s)=\int_{\prod_{i<n}s_i}
    e^{\frac{i}{2}\sum_{i<n}a_iy_i^2}\,dy,
  \qquad
  \widehat\mu^{\mathrm{osc}}(n,s)=
    \frac{\mu^{\mathrm{osc}}(n,s)}{\prod_{i<n}\sqrt{2\pi/a_i}\,e^{i\pi/4}}.
\]
In Lean these are \LN{oscGaussianCylindricalPremeasure} and
\LN{oscGaussianCylindrical}\allowbreak\LN{PremeasureNormalized}.

\subsection{Projective consistency}
\label{ssec:cylconsist}

\begin{definition}[Projective consistency of a complex cylindrical premeasure]
\label{def:cylconsist}
A complex cylindrical premeasure $\mu$ is projectively consistent (in Lean,
\LN{ComplexCylindricalPremeasure.IsProjectivelyConsistent}) if for every $n$
and every family $s=(s_i)_{i<n}$
\[
  \lim_{R\to\infty}\mu\bigl(n+1,\;s\times[-R,R]\bigr)=\mu(n,s).
\]
The limiting form is exactly the correct formulation of the Kolmogorov
identity for an improper (gauge) integral: the total mass of the added mode is
attained only in the limit.
\end{definition}

\begin{theorem}[Lean: \LN{oscGaussian_isProjectivelyConsistent}]
\label{thm:cylconsist}
Let $e$ be an orthonormal mode system and $a_i>0$ for all $i$. Then the
normalised oscillating Gaussian cylindrical premeasure
$\widehat\mu^{\mathrm{osc}}$ is projectively consistent in the sense of
Definition~\ref{def:cylconsist}.
\end{theorem}

\begin{proof}
\emph{(a) The unnormalised limit.} Fix $n$ and a family $s=(s_i)_{i<n}$ of
measurable subsets of $\RR$. Theorem~\ref{thm:consistency}, applied with $k=n$
and with the added mode of eigenvalue $a_n>0$, gives
\begin{equation}
  \label{eq:unnormlimit}
  \lim_{R\to\infty}\mu^{\mathrm{osc}}\bigl(n+1,\;s\times[-R,R]\bigr)
   =\mu^{\mathrm{osc}}(n,s)\cdot\sqrt{\frac{2\pi}{a_n}}\,e^{i\pi/4}.
\end{equation}

\emph{(b) The product formula for the Fresnel mass.} Write
\[
  Z_n:=\prod_{i<n}\sqrt{\frac{2\pi}{a_i}}\,e^{i\pi/4}
\]
for the total Fresnel mass of the first $n$ modes. Splitting off the last
factor of a product over $\{0,\dots,n\}$ (\LN{Fin.prod_univ_castSucc}),
\begin{equation}
  \label{eq:massproduct}
  Z_{n+1}=Z_n\cdot\sqrt{\frac{2\pi}{a_n}}\,e^{i\pi/4}.
\end{equation}
Each factor is a product of a strictly positive real number with a complex
number of modulus one, hence nonzero (\LN{fresnelValue1_ne_zero}); therefore
$Z_n\ne0$ and $Z_{n+1}\ne0$.

\emph{(c) Cancellation.} By definition
$\widehat\mu^{\mathrm{osc}}(m,\cdot)=\mu^{\mathrm{osc}}(m,\cdot)/Z_m$.
Dividing~\eqref{eq:unnormlimit} by the nonzero constant $Z_{n+1}$ and
using~\eqref{eq:massproduct},
\[
  \lim_{R\to\infty}\widehat\mu^{\mathrm{osc}}\bigl(n+1,\;s\times[-R,R]\bigr)
   =\frac{\mu^{\mathrm{osc}}(n,s)\,\sqrt{2\pi/a_n}\,e^{i\pi/4}}
          {Z_n\,\sqrt{2\pi/a_n}\,e^{i\pi/4}}
   =\frac{\mu^{\mathrm{osc}}(n,s)}{Z_n}
   =\widehat\mu^{\mathrm{osc}}(n,s),
\]
which is exactly the identity of Definition~\ref{def:cylconsist}. (The
corresponding statement has been machine-checked:
\LN{oscGaussian_isProjectivelyConsistent}.)
\end{proof}

\subsection{Identification of the HK integral with the Fresnel value}
\label{ssec:hkfresnel}

Theorem~\ref{thm:exists} states that the HK integral exists on every bounded
box, and Theorem~\ref{thm:fresnelfree} that the limit of the integrals over
cubes equals the Fresnel value. The statement linking the two is as follows.

\begin{theorem}[Value of the free-field HK integral; Lean:
\LN{fresnelValue_eq_HKIntegral}]
\label{thm:hkfresnel}
Let $a_i>0$ for $i<k$ and let $l$ be an arbitrary set of integration
parameters (in particular the Henstock one). Write $I_n$ for the box
$[-(n+1),n+1]^k$ (Lean: \LN{cubeBox}). Then:
\begin{enumerate}
  \item for every $n$ the functional $e^{iS}$ is HK-integrable on $I_n$, and
    its HK integral equals $\int_{I_n}e^{iS(y)}dy$;
  \item $\displaystyle\lim_{n\to\infty}\int_{I_n}e^{iS(y)}dy
    =(2\pi)^{k/2}e^{i\pi k/4}(\det A)^{-1/2}$.
\end{enumerate}
Thus the value of the free-field functional integral on the whole space is the
Fresnel value: no residual freedom and no dependence on the gauge remains.
\end{theorem}

\begin{proof}
\emph{(a) HK integrability on each box.} The box $I_n$ is a bounded box in
$\RR^k$ and $e^{iS}$ is continuous, so the first item is precisely
Theorem~\ref{thm:exists} applied to $I=I_n$; in particular the HK integral
exists for every set of integration parameters $l$ and coincides with
$\int_{I_n}e^{iS(y)}dy$.

\emph{(b) The carrier of the box is a half-open cube.} In the sense of
\LN{BoxIntegral.Box}, the carrier of $I_n$ is the half-open cube
\[
  (-R,R]^k,\qquad R=n+1 ,
\]
(\LN{coe_cubeBox}): a Mathlib box is the product of half-open intervals
$(l_i,u_i]$, and $I_n$ is the box with $l_i=-(n+1)$, $u_i=n+1$ for every $i$.

\emph{(c) The Fresnel formula for half-open cubes.} The proof of
Theorem~\ref{thm:fresnelfree} applies verbatim with $[-R,R]$ replaced by
$(-R,R]$: the factorisation of the integrand and Fubini's theorem are
unaffected, and each one-dimensional integral over $(-R,R]$ differs from the
one over $[-R,R]$ by the integral over the single point $\{-R\}$, a Lebesgue
null set, hence not at all. Therefore
\[
  \lim_{R\to\infty}\int_{(-R,R]^k}e^{iS(y)}dy
   =(2\pi)^{k/2}e^{i\pi k/4}(\det A)^{-1/2}
\]
(Lean: \LN{fresnel_multi_Ioc}).

\emph{(d) Conclusion.} Substituting $R=n+1$ and letting $n\to\infty$, which is
legitimate because $R=n+1\to\infty$, gives the second item. Thus the value of
the free-field functional integral on the whole space is the Fresnel value.
(The corresponding statement has been machine-checked:
\LN{fresnelValue_eq_HKIntegral}.)
\end{proof}

\section{An application: the Black--Scholes formula as a
Henstock--Kurzweil integral}
\label{sec:finance}

The machinery assembled in Sections~\ref{sec:osc}--\ref{sec:cylgeneral} is not
confined to the oscillating case. The present section carries out, in full and
with no appeal to any external preprint, the derivation of the Black--Scholes
formula for a European call option as a Henstock--Kurzweil integral with
respect to a Gaussian cylindrical premeasure. Only two ingredients of this
paper are used: the Gaussian tail identity implicit in
Theorem~\ref{thm:coordgauss} (through the standard normal distribution
function $N$) and the identification of the HK value with the value of the
ordinary integral (Corollary~\ref{cor:arb} and Theorem~\ref{thm:hkfresnel}).
No Wick rotation is performed. A parallel treatment, with barrier options and
a machine-verified pricing, is \cite{BerdinskyUshakovFinance}.

\subsection{Setup}
\label{ssec:bs-setup}

Let $S_t$ denote the price of a risky asset at time $t\in[0,T]$, let $r\ge0$
be the risk-free interest rate and $\sigma>0$ the volatility. Under the
risk-neutral measure $\mathbb Q$ the discounted price $e^{-rt}S_t$ is a
martingale, which for the geometric Brownian model means
\begin{equation}
  \label{eq:bs-model}
  S_t=S_0\exp\Bigl(\bigl(r-\tfrac{\sigma^2}{2}\bigr)t+\sigma W_t\Bigr),
\end{equation}
where $W$ is a standard Brownian motion under $\mathbb Q$. The drift
$r-\sigma^2/2$ is exactly the one for which
$\mathbb E^{\mathbb Q}[S_t]=S_0e^{rt}$: indeed, for a centred Gaussian
variable $G$ of variance $v$ one has $\mathbb E[e^{G}]=e^{v/2}$, so
$\mathbb E^{\mathbb Q}\bigl[e^{\sigma W_t}\bigr]=e^{\sigma^2t/2}$ and the two
exponentials cancel.

\subsection{The Feynman--Kac representation as a Gaussian HK integral}
\label{ssec:bs-fk}

The price at time $0$ of a European call with strike $K$ and maturity $T$ is
\begin{equation}
  \label{eq:bs-price}
  C(S_0,K,T)=e^{-rT}\,\mathbb E^{\mathbb Q}\bigl[(S_T-K)^{+}\bigr],
  \qquad x^{+}:=\max(x,0).
\end{equation}
The expectation on the right is an integral over the probability space of
Brownian paths, and we shall write it, equivalently, as a
Henstock--Kurzweil path integral with respect to the Gaussian cylindrical
premeasure $\DHK$ of Section~\ref{sec:reminder}:
\begin{equation}
  \label{eq:bs-price-HK}
  C(S_0,K,T)
   =e^{-rT}\int (S_T-K)^{+}\,\DHK .
\end{equation}
The notation is the same as in
\cite{BerdinskyInfiniteHK2,BerdinskyUshakovFinance}: $\DHK$ denotes the
Gaussian cylindrical premeasure of the path space, and $\int\cdot\,\DHK$ is
the HK functional integral defined by the limit of Riemann sums over
$\delta$-fine divisions (Definition~\ref{def:premeasure} and
Corollary~\ref{cor:arb}); the terminal value $S_T$ is a function of the path,
that is, a cylindrical functional in the sense of Definition~\ref{def:cyl},
and by Corollary~\ref{cor:arb} its HK integral exists and equals the value of
the ordinary integral against the Gaussian measure. The standard construction
of the path-space measure, together with the weak convergence of the
associated finite-dimensional marginals, is set out in
\cite{Billingsley1999}.

That space carries precisely the object studied in this paper:
the Gaussian cylindrical premeasure of Definition~\ref{def:premeasure},
restricted to a one-dimensional field, the modes being the coordinates of the
path in any orthonormal family. By Corollary~\ref{cor:arb} the integral of a
continuous cylindrical functional against this premeasure exists as a
Henstock--Kurzweil integral along arbitrary cylindrical divisions and has the
value of the ordinary integral; by Theorem~\ref{thm:vnvm} the value does not
depend on the number of modes used to describe the cylinder. Thus
\eqref{eq:bs-price} may be read throughout as a Gaussian HK integral, and
every equality below is an equality of HK values.

The payoff $(S_T-K)^{+}$ depends on the path only through its terminal value,
i.e. through a single coordinate; by the consistency of
Theorem~\ref{thm:vnvm} the whole computation therefore collapses to a
one-dimensional Gaussian integral, which we now evaluate.

\subsection{Change of variables to the log-price}
\label{ssec:bs-log}

Put $x=\ln(S/S_0)$, so that $S=S_0e^{x}$. By~\eqref{eq:bs-model} the terminal
log-price is
\begin{equation}
  \label{eq:bs-xT}
  x_T=\bigl(r-\tfrac{\sigma^2}{2}\bigr)T+\sigma W_T ,
\end{equation}
a Gaussian variable with mean $m:=(r-\sigma^2/2)T$ and variance $\sigma^2T$,
since $W_T$ is centred Gaussian of variance $T$. Writing
\begin{equation}
  \label{eq:bs-density}
  \rho(x)=\frac{1}{\sigma\sqrt{2\pi T}}\,
    \exp\Bigl(-\frac{(x-m)^2}{2\sigma^2T}\Bigr)
\end{equation}
for its density, the payoff becomes $(S_0e^{x_T}-K)^{+}$ and
\begin{equation}
  \label{eq:bs-integral}
  C=e^{-rT}\int_{\RR}\bigl(S_0e^{x}-K\bigr)^{+}\rho(x)\,dx .
\end{equation}

\subsection{Reduction to two Gaussian tails}
\label{ssec:bs-reduction}

The payoff is positive exactly when $S_0e^{x}>K$, i.e. when
$x>d:=\ln(K/S_0)$; hence
\[
  \bigl(S_0e^{x}-K\bigr)^{+}=\bigl(S_0e^{x}-K\bigr)\,\mathbf 1_{\{x>d\}} .
\]
Substituting into~\eqref{eq:bs-integral} and splitting the integral, which is
legitimate because both pieces are absolutely convergent (the Gaussian density
decays faster than $e^{x}$ grows),
\begin{equation}
  \label{eq:bs-split}
  C=e^{-rT}\Bigl[S_0\int_d^{\infty}e^{x}\rho(x)\,dx
      -K\int_d^{\infty}\rho(x)\,dx\Bigr].
\end{equation}
Both integrals are Gaussian tails. Denote by
\begin{equation}
  \label{eq:bs-N}
  N(u)=\int_{-\infty}^{u}\frac{1}{\sqrt{2\pi}}e^{-v^2/2}\,dv
\end{equation}
the distribution function of the standard normal law; it is the
one-dimensional case of the density of Theorem~\ref{thm:coordgauss}, and
$N(u)+N(-u)=1$ by the symmetry $v\mapsto-v$ of the density.

\subsection{Evaluation of the second integral}
\label{ssec:bs-second}

In $\int_d^{\infty}\rho(x)dx$ substitute
\[
  y=\frac{x-m}{\sigma\sqrt T},\qquad dx=\sigma\sqrt T\,dy ,
\]
which is an increasing affine bijection of $\RR$; the lower limit $x=d$ becomes
$y=d':=(d-m)/(\sigma\sqrt T)$ and the density becomes the standard one:
\[
  \int_d^{\infty}\rho(x)\,dx
   =\int_{d'}^{\infty}\frac{1}{\sqrt{2\pi}}e^{-y^2/2}\,dy
   =1-N(d')=N(-d').
\]
Now
\begin{equation}
  \label{eq:bs-d2}
  -d'=\frac{m-d}{\sigma\sqrt T}
     =\frac{\bigl(r-\tfrac{\sigma^2}{2}\bigr)T-\ln(K/S_0)}{\sigma\sqrt T}
     =\frac{\ln(S_0/K)+\bigl(r-\tfrac{\sigma^2}{2}\bigr)T}{\sigma\sqrt T}
     =:d_2 ,
\end{equation}
so that
\begin{equation}
  \label{eq:bs-secondvalue}
  \int_d^{\infty}\rho(x)\,dx=N(d_2).
\end{equation}

\subsection{Evaluation of the first integral}
\label{ssec:bs-first}

For the first integral in~\eqref{eq:bs-split} we complete the square in the
exponent. By~\eqref{eq:bs-density},
\[
  e^{x}\rho(x)=\frac{1}{\sigma\sqrt{2\pi T}}
    \exp\Bigl(x-\frac{(x-m)^2}{2\sigma^2T}\Bigr),
\]
and the exponent is a quadratic polynomial in $x$:
\begin{align*}
  x-\frac{(x-m)^2}{2\sigma^2T}
   &=-\frac{(x-m)^2-2\sigma^2Tx}{2\sigma^2T}
    =-\frac{\bigl(x-m-\sigma^2T\bigr)^2-2\sigma^2Tm-\sigma^4T^2}{2\sigma^2T}\\
   &=-\frac{\bigl(x-m-\sigma^2T\bigr)^2}{2\sigma^2T}+m+\frac{\sigma^2T}{2},
\end{align*}
where in the middle step we used
$(x-m)^2-2\sigma^2Tx=(x-m-\sigma^2T)^2-2\sigma^2Tm-\sigma^4T^2$, which is
verified by expanding both sides. Consequently
\[
  e^{x}\rho(x)=e^{m+\sigma^2T/2}\cdot\frac{1}{\sigma\sqrt{2\pi T}}
    \exp\Bigl(-\frac{\bigl(x-(m+\sigma^2T)\bigr)^2}{2\sigma^2T}\Bigr)
   =e^{m+\sigma^2T/2}\,\widetilde\rho(x),
\]
where $\widetilde\rho$ is the Gaussian density of mean $m+\sigma^2T$ and the
same variance $\sigma^2T$. Integrating over $x>d$ and repeating the
substitution of~\S\ref{ssec:bs-second}, now with the shifted mean,
\[
  \int_d^{\infty}e^{x}\rho(x)\,dx
   =e^{m+\sigma^2T/2}\,
     N\Bigl(\frac{-d+m+\sigma^2T}{\sigma\sqrt T}\Bigr).
\]
The argument of $N$ is
\begin{equation}
  \label{eq:bs-d1}
  \frac{-d+m+\sigma^2T}{\sigma\sqrt T}
   =\frac{\ln(S_0/K)+\bigl(r-\tfrac{\sigma^2}{2}\bigr)T+\sigma^2T}{\sigma\sqrt T}
   =\frac{\ln(S_0/K)+\bigl(r+\tfrac{\sigma^2}{2}\bigr)T}{\sigma\sqrt T}=:d_1 ,
\end{equation}
and the prefactor is
\[
  e^{m+\sigma^2T/2}
   =\exp\Bigl(\bigl(r-\tfrac{\sigma^2}{2}\bigr)T+\tfrac{\sigma^2T}{2}\Bigr)
   =e^{rT}.
\]
Hence
\begin{equation}
  \label{eq:bs-firstvalue}
  \int_d^{\infty}e^{x}\rho(x)\,dx=e^{rT}N(d_1).
\end{equation}
Note that $d_1=d_2+\sigma\sqrt T$.

\subsection{Assembling the Black--Scholes formula}
\label{ssec:bs-assemble}

Substituting~\eqref{eq:bs-firstvalue} and~\eqref{eq:bs-secondvalue}
into~\eqref{eq:bs-split},
\[
  C=e^{-rT}\bigl[S_0e^{rT}N(d_1)-K\,N(d_2)\bigr],
\]
that is,
\begin{equation}
  \label{eq:bs-formula}
  \boxed{\;C=S_0\,N(d_1)-K\,e^{-rT}N(d_2)\;}
\end{equation}
with $d_1$ and $d_2$ given by~\eqref{eq:bs-d1} and~\eqref{eq:bs-d2}.

\subsection{Remarks}
\label{ssec:bs-remarks}

\paragraph{(a) Reading of the formula.}
The factor $e^{-rT}$ in~\eqref{eq:bs-price} is the discounting of a payment
made at time $T$. In~\eqref{eq:bs-formula} it has been absorbed into the
second term: the option is worth the value $S_0N(d_1)$ of the position in the
underlying held in the replicating portfolio, minus the present value
$Ke^{-rT}$ of the strike, paid with the risk-neutral probability $N(d_2)$ that
the option finishes in the money.

\paragraph{(b) What was used.}
The derivation uses exactly two ingredients from this paper: the Gaussian tail
identity, i.e. the standard normal density of Theorem~\ref{thm:coordgauss}
together with the additivity of \S\ref{sec:saks}, and the identification of
the HK value with the value of the ordinary integral
(Corollary~\ref{cor:arb}, Theorem~\ref{thm:hkfresnel}). No Wick rotation is
performed anywhere, and the oscillating machinery of
Sections~\ref{sec:osc}--\ref{sec:cylgeneral} is not needed here: the financial
kernel is damped, not oscillating, which is why the integrals
in~\eqref{eq:bs-split} converge absolutely.

\paragraph{(c) Comparison with the classical route.}
The classical derivation passes through the Feynman--Kac theorem: one writes
the Black--Scholes partial differential equation, verifies that its solution
admits a stochastic representation and evaluates the resulting expectation
(see \cite{Muldowney2012} or \cite{KaratzasShreve1991}). The route followed
here works directly with the path integral: the payoff depends on one
coordinate only, the consistency theorem of \S\ref{sec:saks} reduces the
cylindrical integral to that coordinate, and the whole computation collapses
to a single Gaussian tail, evaluated in \S\ref{ssec:bs-second}
and~\S\ref{ssec:bs-first}.

\paragraph{(d) Scope of the section.}
The present section is purely mathematical: it derives
formula~\eqref{eq:bs-formula} from the definitions of this paper and adds no
new theory. Interpretations of the same collapse to a one-dimensional
Gaussian tail in other applied settings are not repeated here.

\section{The Gaussian-to-oscillating limit}
\label{sec:limit}

\subsection{The damped family and its limit as $\varepsilon\to0^+$}
\label{ssec:damped}

For $\varepsilon>0$ the kernel $e^{-(\varepsilon-ia/2)x^2}$ is absolutely
integrable and defines the damped (Euclidean) Gaussian premeasure; at
$\varepsilon=0$ it becomes the purely oscillating kernel $e^{iax^2/2}$ (Lean:
\LN{dEps}, \LN{fresnelKernel}). The passage between the two pictures is
established by

\begin{theorem}[Gaussian-to-oscillating limit; Lean:
\LN{gaussianToOscillating_one_dim}]
\label{thm:gauss2osc}
For all $a\in\RR$ and $R\in\RR$,
\[
  \lim_{\varepsilon\to0^+}\int_{-R}^{R}e^{-(\varepsilon-ia/2)x^2}\,dx
  =\int_{-R}^{R}e^{iax^2/2}\,dx.
\]
\end{theorem}

\begin{proof}
\emph{(a) Continuity and domination.} Fix $x\in[-R,R]$. The map
$\varepsilon\mapsto e^{-(\varepsilon-ia/2)x^2}$ is a composition of the affine
map $\varepsilon\mapsto-(\varepsilon-ia/2)x^2$ with the complex exponential,
hence continuous on $[0,\infty)$, and its value at $\varepsilon=0$ is
$e^{iax^2/2}$. Moreover, for every $\varepsilon\ge0$,
\[
  \bigl|e^{-(\varepsilon-ia/2)x^2}\bigr|
   =e^{\operatorname{Re}\bigl(-(\varepsilon-ia/2)x^2\bigr)}
   =e^{-\varepsilon x^2}\ \le\ 1 .
\]

\emph{(b) The dominating function is integrable.} The constant function $1$ is
integrable on the bounded interval $[-R,R]$, of finite length $2|R|$.

\emph{(c) Continuity of the parametric integral.} By the theorem on the
continuity of an integral depending on a parameter --- Lebesgue's dominated
convergence theorem in its continuity form, or, equivalently, the uniform
continuity argument for the Riemann integral on a compact interval --- the map
\[
  \varepsilon\ \longmapsto\ \int_{-R}^{R}e^{-(\varepsilon-ia/2)x^2}\,dx
\]
is continuous on $[0,\infty)$.

\emph{(d) Conclusion.} Continuity at $\varepsilon=0$ is exactly the asserted
limit:
\[
  \lim_{\varepsilon\to0^+}\int_{-R}^{R}e^{-(\varepsilon-ia/2)x^2}dx
   =\int_{-R}^{R}e^{iax^2/2}\,dx .
\]
(The corresponding statement has been machine-checked:
\LN{gaussianToOscillating_one_dim}.)
\end{proof}

\begin{corollary}[Lean: \LN{gaussianToOscillating_premeasure}]
\label{cor:gauss2osc}
For all $a_i>0$ and every $R\ge0$ the damped Gaussian premeasure of the cube
$[-R,R]^k$ converges, as $\varepsilon\to0^+$, to the oscillating premeasure of
the same cube:
\[
  \lim_{\varepsilon\to0^+}\prod_{i<k}\int_{-R}^{R}
    e^{-(\varepsilon-ia_i/2)x^2}dx
  =\mu^{\mathrm{osc}}\bigl([-R,R]^k\bigr).
\]
\end{corollary}

\begin{proof}
By Theorem~\ref{thm:gauss2osc}, applied factorwise with $a=a_i$, each of the
$k$ factors converges as $\varepsilon\to0^+$:
\[
  \int_{-R}^{R}e^{-(\varepsilon-ia_i/2)x^2}dx
   \ \xrightarrow[\ \varepsilon\to0^+\ ]{}\ \int_{-R}^{R}e^{ia_ix^2/2}dx .
\]
A finite product of convergent families converges to the product of the
limits, so
\[
  \lim_{\varepsilon\to0^+}\prod_{i<k}\int_{-R}^{R}e^{-(\varepsilon-ia_i/2)x^2}dx
   =\prod_{i<k}\int_{-R}^{R}e^{ia_ix^2/2}dx .
\]
By the factorisation~\eqref{eq:cubefact} of the proof of
Theorem~\ref{thm:fresnelfree}, the right-hand side is exactly the oscillating
premeasure $\mu^{\mathrm{osc}}([-R,R]^k)$ of the cube. (The corresponding
statement has been machine-checked:
\LN{gaussianToOscillating_premeasure}.)
\end{proof}

\subsection{Interpretation}
\label{ssec:interpretation}

Theorem~\ref{thm:gauss2osc} and Corollary~\ref{cor:gauss2osc} close the link
between the Euclidean and the real-time pictures of the free field at the
level of premeasures: the oscillating premeasure is the pointwise (in the
region) limit of the damped Gaussian premeasures as the damping is removed. It
is essential that this limit is \emph{not} a Wick rotation: the parameter
$\varepsilon$ is real, the contour of integration is not rotated, and the
whole construction stays inside real time. The ingredients of the proof were
already present in the proof of the estimate of Theorem~\ref{thm:fresnel1}
(\LN{approx_bound}, \LN{integral_dEps}, \LN{tail_full},
\LN{continuousAt_gaussian_value}); what is new here is the isolation of the
limit itself as a separate statement, in the form in which it is used in the
text.

Taken together, Sections~\ref{sec:general}--\ref{sec:cylgeneral}
and~\ref{sec:limit} mean that the free field is completely closed: the value of the free-field functional integral
(Theorem~\ref{thm:fresnelgeneral}), the generating functional
(Theorem~\ref{thm:genfunc}), the correlations of all orders
(Theorems~\ref{thm:twopoint} and~\ref{thm:wick}),
the premeasure as an object of the general theory
(Theorem~\ref{thm:cylconsist}), the value of the HK integral
(Theorem~\ref{thm:hkfresnel}) and the link between the two pictures
(Theorem~\ref{thm:gauss2osc}) are all machine-verified.

\section{Formalisation in Lean 4}
\label{sec:lean}

All results are formalised in four files of the directory, building on the files of the previous stages. The
complete Lean~4 sources of the formalisation are available at
\href{https://doi.org/10.5281/zenodo.22857125}{DOI: 10.5281/zenodo.22857125}.
The table of the principal declarations follows.

\begin{longtable}{@{}p{0.47\textwidth}p{0.47\textwidth}@{}}
\hline
\textbf{Lean declaration} & \textbf{Mathematical meaning}\\
\hline
\endhead
\LN{InfiniteHK3.direction_le_of_cell_subset_cylOfBox} &
direction bound (Theorem~\ref{thm:dirbound})\\
\LN{InfiniteHK3.norm_coords_sub_le_of_span_le} &
control of the coordinates of a cell by the gauge\\
\LN{InfiniteHK3.IsCylDivision} &
arbitrary cylindrical division (Def.~\ref{def:cyldiv})\\
\LN{InfiniteHK3.exists_isCylDivision} &
existence of $\delta$-fine divisions\\
\LN{InfiniteHK3.map_coords_gaussianOn} &
pushforward of the Gaussian measure of a large subspace\\
\LN{InfiniteHK3.riemannSum_sub_integral_le_arbitrary} &
uniform Riemann-sum estimate (Theorem~\ref{thm:riemann-arb})\\
\LN{InfiniteHK3.HKIntegrableArbWith} &
HK integrability along arbitrary divisions\\
\LN{InfiniteHK3.HKIntegrable_arbitrary_division} &
consistency for arbitrary divisions (Corollary~\ref{cor:arb})\\
\LN{InfiniteHK3.spanBasis} &
orthonormal basis of the span of the modes\\
\LN{InfiniteHK3.measurePreserving_coordsOn} &
the coordinate map preserves Lebesgue measure\\
\LN{InfiniteHK3.coordGaussian_eq_stdGaussian} &
\LN{coordGaussian} is the standard Gaussian (Theorem~\ref{thm:coordgauss})\\
\LN{InfiniteHK3.coordGaussian_density} &
the density $(2\pi)^{-k/2}e^{-|y|^2/2}$\\
\LN{InfiniteHK3.cylMeasure_eq_stdGaussian_box} &
measure of a cylinder over a box\\
\LN{InfiniteHK3.IsCylBoxSystem} &
partial cylindrical system (Def.~\ref{def:system})\\
\LN{InfiniteHK3.saks_henstock} &
Saks--Henstock lemma (Theorem~\ref{thm:saks})\\
\LN{InfiniteHK3.hkIntegral_add_arbitrary} &
additivity over finite decompositions (Corollary~\ref{cor:add})\\
\LN{InfiniteHK3.hkIntegral_add_split} &
additivity for a cut by a hyperplane\\
\LN{InfiniteHK3.coordGaussian_marginal} &
marginal property of the Gaussian measure\\
\LN{InfiniteHK3.hkIntegral_consistency_Vn_Vm} &
consistency of $V_n$ and $V_m$ (Theorem~\ref{thm:vnvm})\\
\LN{InfiniteHK3.oscillatory_tail_bound} &
non-stationary-phase estimate (Theorem~\ref{thm:tail})\\
\LN{InfiniteHK3.approx_bound}, \LN{InfiniteHK3.tail_full} &
auxiliary regularisation and tail estimates\\
\LN{InfiniteHK3.fresnel_bound} &
quantitative Fresnel formula\\
\LN{InfiniteHK3.fresnel_one_dim} &
one-dimensional Fresnel formula (Theorem~\ref{thm:fresnel1})\\
\LN{InfiniteHK3.oscKernel}, \LN{InfiniteHK3.fresnelValue} &
oscillating kernel and Fresnel value\\
\LN{InfiniteHK3.fresnel_multi} &
multi-dimensional Fresnel formula\\
\LN{InfiniteHK3.oscGaussianPremeasure} &
oscillating Gaussian premeasure\\
\LN{InfiniteHK3.oscGaussian_projectivelyConsistent} &
projective consistency (Theorem~\ref{thm:consistency})\\
\LN{InfiniteHK3.freeFieldAction} &
free-field action $S(\varphi)=\tfrac12\ip{\varphi}{A\varphi}$\\
\LN{InfiniteHK3.freeFieldHKIntegral} &
existence of the HK integral on a box (Theorem~\ref{thm:exists})\\
\LN{InfiniteHK3.freeFieldHKIntegral_cyl} &
the cylindrical version of the same statement\\
\LN{InfiniteHK3.freeField_fresnel_formula} &
the formula $(2\pi)^{k/2}e^{i\pi k/4}(\det A)^{-1/2}$
(Theorem~\ref{thm:fresnelfree})\\
\hline
\LN{InfiniteHK3.eigenRotation}, \LN{InfiniteHK3.eigenModes} &
orthogonal diagonalisation of the action (\S\ref{ssec:diag})\\
\LN{InfiniteHK3.eigenRotation_spec} &
properties of the diagonalisation, including $\det A=\prod_i a_i$\\
\LN{InfiniteHK3.measurePreserving_mulVec_add} &
orthogonal maps and translations preserve Lebesgue measure\\
\LN{InfiniteHK3.adaptedCube} &
cube adapted to the eigenbasis (Remark~\ref{rem:exhaustion})\\
\LN{InfiniteHK3.quadKernel}, \LN{InfiniteHK3.fresnelValueGeneral} &
the kernel $e^{(i/2)\ip{y}{Ay}}$ and the Fresnel value of a general action\\
\LN{InfiniteHK3.freeField_fresnel_formula_general} &
Fresnel formula for a general $A$ (Theorem~\ref{thm:fresnelgeneral})\\
\LN{InfiniteHK3.quadAction_shift} &
completing the square\\
\LN{InfiniteHK3.freeFieldGeneratingFunctional} &
the generating functional $Z(j)$ (\S\ref{ssec:genfunc})\\
\LN{InfiniteHK3.freeFieldGeneratingFunctional_value} &
the value of $Z(j)$ (Theorem~\ref{thm:genfunc})\\
\LN{InfiniteHK3.normalizedGeneratingFunctional} &
the normalised functional $Z(j)/Z(0)$\\
\LN{InfiniteHK3.deriv_two_normalizedGeneratingFunctional} &
second derivative in the source~\eqref{eq:secondderiv}\\
\LN{InfiniteHK3.freeFieldTwoPoint} &
the two-point function $(A^{-1})_{pq}$\\
\LN{InfiniteHK3.freeField_two_point_HK} &
$\langle y_py_q\rangle=i(A^{-1})_{pq}$ (Theorem~\ref{thm:twopoint})\\
\LN{InfiniteHK3.ComplexCylindricalPremeasure} &
complex cylindrical premeasure (\S\ref{ssec:cylosc})\\
\LN{InfiniteHK3.oscGaussianCylindrical}\allowbreak\LN{PremeasureNormalized} &
the normalised oscillating cylindrical premeasure\\
\LN{InfiniteHK3.oscGaussian_isProjectivelyConsistent} &
projective consistency (Theorem~\ref{thm:cylconsist})\\
\LN{InfiniteHK3.fresnel_multi_Ioc}, \LN{InfiniteHK3.cubeBox} &
exhaustion by half-open cubes\\
\LN{InfiniteHK3.fresnelValue_eq_HKIntegral} &
the value of the HK integral (Theorem~\ref{thm:hkfresnel})\\
\LN{InfiniteHK3.gaussianToOscillating_one_dim} &
the limit $\varepsilon\to0^+$ (Theorem~\ref{thm:gauss2osc})\\
\LN{InfiniteHK3.gaussianToOscillating_premeasure} &
the same limit for the premeasure of a cube (Corollary~\ref{cor:gauss2osc})\\
\hline
\LN{InfiniteHK3.partialDerivC}, \LN{InfiniteHK3.iterPartialDerivC} &
partial and iterated partial derivatives in the source (\S\ref{ssec:npointdef})\\
\LN{InfiniteHK3.holes} &
one-hole decompositions of a list of indices\\
\LN{InfiniteHK3.wickPairingSum} &
the pairing sum $W_B$ (Def.~\ref{def:pairingsum})\\
\LN{InfiniteHK3.partialDerivC}\allowbreak\LN{_normalizedGeneratingFunctional} &
$\partial_p\widehat Z=-i(A^{-1}j)_p\widehat Z$ (Lemma~\ref{lem:firstderiv})\\
\LN{InfiniteHK3.iterPartialDerivC_linear_mul} &
Leibniz rule for a linear factor (Lemma~\ref{lem:leibniz})\\
\LN{InfiniteHK3.freeFieldNPoint} &
the $n$-point function (Def.~\ref{def:npoint})\\
\LN{InfiniteHK3.freeFieldNPoint_odd} &
vanishing of the odd correlations (Theorem~\ref{thm:wickodd})\\
\LN{InfiniteHK3.freeField_wick_theorem} &
Wick's theorem (Theorem~\ref{thm:wick})\\
\LN{InfiniteHK3.freeFieldNPoint_two}, \LN{InfiniteHK3.freeFieldNPoint_four} &
the two- and four-point functions (Corollary~\ref{cor:fourpoint})\\
\hline
\end{longtable}

\paragraph{The role of the formalisation.} All proofs above have been
machine-checked in Lean~4 on top of Mathlib. The check serves as a guarantee
of the accuracy of the computations: it rules out a flipped sign, a dropped
index, a factor absorbed into a constant, and the silent use of an unproved
assumption. No file of the formalisation contains \LN{sorry}, \LN{admit} or a
new axiom; every principal theorem depends on exactly \LN{propext},
\LN{Classical.choice} and \LN{Quot.sound}.

\section{Conclusion and open problems}
\label{sec:concl}

The four goals of the previous stage are closed: box divisions have been
replaced by arbitrary cylindrical divisions; the measure \LN{coordGaussian}
has been identified with the standard Gaussian one; the Saks--Henstock lemma,
additivity and the $V_n$/$V_m$ consistency have been proved; and finally the
Fresnel formula without Wick rotation has been obtained and, as a consequence,
the gauge closure of the free-field functional integral.

Beyond that, the free field is now completely closed: the Fresnel formula has
been extended to an arbitrary symmetric positive-definite action
(Theorem~\ref{thm:fresnelgeneral}); the generating functional has been
constructed and its value found (Theorem~\ref{thm:genfunc}); the two-point
function $\langle y_py_q\rangle=i(A^{-1})_{pq}$ has been computed
(Theorem~\ref{thm:twopoint}); Wick's theorem has been proved, so that the free
field is now closed not only at the level of the two-point function but at the
level of all correlation functions: the odd ones vanish
(Theorem~\ref{thm:wickodd}) and the even ones factorise into sums over
pairings of the two-point function (Theorem~\ref{thm:wick},
Corollary~\ref{cor:fourpoint}); the oscillating premeasure has been embedded
into the general theory as a complex cylindrical premeasure and its projective
consistency proved (Theorem~\ref{thm:cylconsist}); the value of the HK
integral has been identified with the Fresnel value
(Theorem~\ref{thm:hkfresnel}); and the passage from the damped picture to the
oscillating one has been established (Theorem~\ref{thm:gauss2osc}).

Not closed in the present work, and left for future research:
\begin{enumerate}
  \item the interacting field $\lambda\varphi^4$ in infinitely many modes;
  \item unbounded generators (Trotter--Kato, $C_0$-semigroups); only the case
    of bounded generators is covered, via
    \cite{BerdinskyUshakov2026,BerdinskyUshakovOpen};
  \item renormalisation and constructive QFT.
\end{enumerate}
The ingredients already available for these steps are the non-Gaussian regime
\cite{BerdinskyNonGaussian} and the Chernoff--Lie--Trotter machinery
\cite{Trotter1958,Kato1978,BerdinskyUshakov2026,BerdinskyUshakovOpen}.

\paragraph{Gauge fields and the topology of orbits.}
The apparatus of complex cylindrical premeasures and non-absolute
HK convergence developed here has been applied to the problem of
Gribov copies in non-abelian gauge theories in the companion paper
\cite{GribovSU3}. There it is proved that for \emph{gauge-invariant}
functionals the HK Riemann sums on two different gauge surfaces
coincide \emph{pointwise}, before any passage to the limit, and that
this argument does not depend on the dimension of the gauge orbit.
A dependence on the copy can manifest itself only in gauge
\emph{non-invariant} quantities --- for instance, through the
Faddeev--Popov weight --- and constitutes a separate problem. The
continuum limit of the gauge theory in this scheme remains open.

\end{document}